\documentclass[
  aps,
  prx,
  reprint,
  superscriptaddress,
  amsmath,
  amssymb,
  longbibliography,
  floatfix
]{revtex4-2}

\usepackage{amsthm}
\usepackage{graphicx}
\usepackage{tikz}
\usetikzlibrary{matrix}
\newcommand{\partitionedmatrix}[1]{%
  \begin{tikzpicture}[baseline=-0.5ex]
    \matrix (m) [matrix of math nodes, ampersand replacement=\&,
      left delimiter=(, right delimiter=), inner sep=0pt,
      nodes={inner sep=0pt, minimum width=1em, minimum height=1.2em}] {#1};
    \draw[dash pattern=on 2pt off 2pt] (m-1-4.north east) -- (m-8-4.south east);
    \draw[dash pattern=on 2pt off 2pt] (m-4-1.south west) -- (m-4-8.south east);
  \end{tikzpicture}%
}
\usepackage{hyperref}

\hypersetup{
  colorlinks=true,
  citecolor=blue,
  linkcolor=blue,
  urlcolor=blue
}

\newtheorem{theorem}{Theorem}
\newtheorem{lemma}[theorem]{Lemma}
\newtheorem{definition}[theorem]{Definition}

\newcommand{\F}{\mathbb{F}}
\newcommand{\cH}{\mathcal{H}}
\newcommand{\cU}{\mathcal{U}}
\newcommand{\Sp}{\operatorname{Sp}}
\newcommand{\rank}{\operatorname{rank}}
\newcommand{\rad}{\operatorname{rad}}
\newcommand{\OpSch}{\operatorname{OpSch}}
\newcommand{\SWAP}{\operatorname{SWAP}}
\newcommand{\CZ}{\operatorname{CZ}}

\begin{document}

\title{Reducing the Entanglement Cost of Distributed Bipartite Quantum Computation with Constant Qubit Overhead}

\author{Kosuke Matsui}
\email{kosuke.matsui@phys.s.u-tokyo.ac.jp}
\affiliation{Department of Physics, Graduate School of Science, The University of Tokyo, Hongo 7-3-1, Bunkyo-ku, Tokyo, Japan}
\affiliation{Hon Hai (Foxconn) Research Institute, Taipei, Taiwan}

\author{Jun-Yi Wu}
\email{junyiwuphysics@gmail.com}
\affiliation{Hon Hai (Foxconn) Research Institute, Taipei, Taiwan}
\affiliation{Department of Physics, Tamkang University, New Taipei 251301, Taiwan, Republic of China}

\author{Min-Hsiu Hsieh}
\email{min-hsiu.hsieh@foxconn.com}
\affiliation{Hon Hai (Foxconn) Research Institute, Taipei, Taiwan}

\author{Mio Murao}
\email{murao@phys.s.u-tokyo.ac.jp}
\affiliation{Department of Physics, Graduate School of Science, The University of Tokyo, Hongo 7-3-1, Bunkyo-ku, Tokyo, Japan}

\begin{abstract}
Distributed quantum computation connects multiple quantum processing units (QPUs) through quantum communication to jointly perform large-scale quantum computations.
Since the number of qubits available at each QPU is limited, it is important to reduce quantum communication while keeping the qubit overhead small.
To this end, we study an entanglement-assisted model, in which entanglement consumption serves as a measure of quantum communication cost.
For exact deterministic implementations in this model, the operator Schmidt rank of the target bipartite unitary provides a general lower bound on entanglement cost when qubit overhead is unrestricted.
However, it has remained unclear how closely this bound can be approached with constant qubit overhead.
In this work, we show that this lower bound is attainable for every bipartite Clifford unitary using \emph{at most two auxiliary qubits per QPU}.
In addition, for non-Clifford unitaries specified by exact Clifford+$T$ decompositions with $T$-count $t$, we construct implementations with the same qubit overhead whose entanglement cost is at most $2t$ Bell pairs above this lower bound.
\end{abstract}

\maketitle

\section{Introduction}

Scaling the number of qubits available in a single quantum processor remains technically challenging~\cite{Bravyi2024,2023google,mohseni2024,zhou2025}.
Distributed quantum computation connects smaller quantum processing units (QPUs) through quantum communication so that they can execute computations that require more qubits than are available in any single QPU~\cite{eisert2000,wu2023,andresmartinez2024,main2025}.
In this work, we consider an entanglement-assisted model in which nonlocal gates between QPUs are implemented using local operations and classical communication (LOCC) assisted by shared entanglement~\cite{bennett1993,eisert2000,Collins2001,Cirac2001,DurCirac2001,Devetak2004,Devetak2008}.
In this model, it is important to reduce the entanglement cost, the amount of shared entanglement consumed during an implementation.
Regarding qubit requirements, distributing a computation across multiple QPUs can require more qubits in total than executing it on a single quantum processor, resulting in qubit overhead.
These additional qubits are used, for example, to hold the local halves of shared Bell pairs or quantum states teleported from another QPU.
Since the number of qubits available at each QPU is limited, it is important to reduce the entanglement cost while keeping the qubit overhead small.

Existing circuit-compilation methods pursue this goal through several approaches.
Some optimize the placement of qubits among QPUs to reduce quantum communication while respecting each QPU's qubit limit~\cite{andresmartinez2019,sundaram2021,Sundaram_2022,Ferrari_2023,Burt_2024}.
Another approach reduces the entanglement cost of distributed computations by implementing suitable groups of nonlocal gates together.
The nonlocal gates are grouped so that the implementation requires no more auxiliary qubits than are available at each QPU~\cite{autocomm2022,wu2023,andresmartinez2024}.
The resulting implementations reduce entanglement costs under the respective qubit constraints.
However, these studies do not establish whether the entanglement costs they achieve are optimal or, if not, how much they exceed the minimum possible entanglement cost.

Without constraints on qubit overhead, lower bounds on entanglement cost and their attainability have been studied in various operational settings~\cite{eisert2000,Chen_PhysRevA.93.042331,Soeda_PhysRevLett.107.180501,Wakakuwa_CodingThorem,Wakakuwa_ComplexityCausal,akibue2025computingentanglementcostsnonlocal,cleve2026lowerboundsnonlocalcomputation}.
For exact deterministic entanglement-assisted LOCC implementations of a bipartite unitary, the operator Schmidt rank of the unitary provides a lower bound on the entanglement cost~\cite{stahlke2011}.
In particular, this bound can be attained for Clifford unitaries, which play a central role in fault-tolerant quantum computation~\cite{gottesmanchuang1999}.
One implementation that attains this bound prepares the Choi state of the target Clifford unitary and uses it for gate teleportation~\cite{fattal2004,gottesmanchuang1999}.
However, storing this Choi state requires auxiliary qubits in proportion to the number of qubits on which the original unitary acts.
This raises the question: how closely can this general lower bound on entanglement cost be approached with constant qubit overhead?

Our first main result shows that this lower bound remains attainable for Clifford unitaries with constant qubit overhead.
Specifically, every bipartite Clifford unitary admits an exact deterministic entanglement-assisted LOCC implementation that attains the bound using at most two auxiliary qubits at each QPU.
In our setting, qubits may be measured during circuit execution and then reset~\cite{Iqbal2024,Ryan-Anderson2021realtimefaulttolerantqec,Gramham2023midcircuitmeasurements,Lis2023midcircuitoperations,PhysRevApplied.17.014014,Sivak_2023,Lloyd2001engineeringquantumdynamics,Andersson2008binarysearchtrees,Ivashkov2024povmswithdynamiccircuits,Shen2017quantumchannelconstruction} or reused to prepare Bell pairs.
Our result relies on decomposing the target Clifford unitary into Clifford gates that can each be implemented with constant qubit overhead, while keeping the total entanglement consumption optimal.
Thus, restricting the qubit overhead to a constant need not increase the entanglement cost of Clifford implementations.

Our second main result bounds the excess entanglement cost above this lower bound for non-Clifford unitaries specified by Clifford+$T$ decompositions~\cite{boykin2000} while keeping qubit overhead constant.
Specifically, we give an exact deterministic entanglement-assisted LOCC implementation of a bipartite unitary specified by an exact Clifford+$T$ decomposition with a $T$-count of $t$.
The implementation uses at most two auxiliary qubits at each QPU, with an entanglement cost at most $2t$ Bell pairs above the lower bound.
This bound on the excess entanglement cost depends only on the $T$-count of the given decomposition, with no additional term depending on the depth or total gate count of the original circuit.

The remainder of this paper is organized as follows.
We describe our problem setting and define qubit overhead in Sec.~\ref{sec:setting}.
Sections~\ref{sec:clifford} and~\ref{sec:clifford-t} present our results for Clifford unitaries and unitaries with Clifford+$T$ decompositions, respectively.
The appendixes provide technical details and proofs.

\section{Problem setting}
\label{sec:setting}

In this section, we specify our setting for the distributed implementation of bipartite unitaries and define the associated resource costs.

Alice and Bob hold $n_A$ and $n_B$ qubits, with Hilbert spaces $\cH_A:=(\mathbb{C}^2)^{\otimes n_A}$ and $\cH_B:=(\mathbb{C}^2)^{\otimes n_B}$.
They implement a bipartite unitary $U\in\cU(\cH_A\otimes\cH_B)$ exactly and deterministically by local operations and classical communication (LOCC), assisted by Bell pairs.
Here, $\cU(\cH)$ denotes the set of unitary operators on a Hilbert space $\cH$.
We assume that all supplied Bell pairs are perfect, each in the state
\begin{equation}
  \lvert\Phi\rangle:=\frac{1}{\sqrt{2}}(\lvert0\rangle\lvert0\rangle+\lvert1\rangle\lvert1\rangle).
  \label{eq:bell-pair}
\end{equation}
In this paper, we define the \emph{entanglement cost} of an implementation as the total number of Bell pairs it consumes.

In addition, measurements may be performed during circuit execution, after which the measured qubits can be reset to $\lvert0\rangle$ or reused to prepare perfect Bell pairs.
With these operations allowed, we use the term \emph{qubit overhead} to quantify the number of qubits required at each party throughout the implementation in addition to those needed to hold the corresponding input state.
A qubit overhead of $(m_A,m_B)$ means that at most $n_A+m_A$ and $n_B+m_B$ qubits are used at Alice and Bob, respectively, at any stage of the implementation.
Note that any additional qubits required to generate Bell pairs are excluded from the qubit overhead.
The following definition formalizes this setting and the associated qubit overhead.

\begin{definition}[Qubit overhead in distributed implementations]
\label{def:space-overhead}
Let $U\in\cU(\cH_A\otimes\cH_B)$ be a bipartite unitary operator, and let $m_A,m_B\in\mathbb{Z}_{\geq0}$.
For a protocol with $R$ rounds, let $i_r$ and $j_r$ denote Alice's and Bob's measurement outcomes in round $r$, respectively.
We write $h_{\le r}:=(i_1,j_1,\ldots,i_r,j_r)$ for the sequence of measurement outcomes through round $r$, with $h_{\le0}:=()$ denoting the empty sequence.

We say that $U$ has an entanglement-assisted LOCC implementation with qubit overhead $(m_A,m_B)$ if, for some finite $R$, there exist unitary operators $V^{(r)}_{h_{\le r-1}}\in\cU((\mathbb{C}^2)^{\otimes(n_A+m_A)})$ and $W^{(r)}_{h_{\le r-1}}\in\cU((\mathbb{C}^2)^{\otimes(n_B+m_B)})$ for each $r=1,2,\ldots,R$ and each sequence of measurement outcomes $h_{\le r-1}$, such that the protocol shown in Figure~\ref{fig:qubit-space-overhead} implements $U$ exactly and deterministically.
More formally, for every normalized input state $\lvert\psi\rangle\in\cH_A\otimes\cH_B$, we require
\begin{equation}
  U\bigl(\lvert\psi\rangle\!\langle\psi\rvert\bigr)U^\dagger
  =\sum_{h_{\le R}}E_{h_{\le R}}
  \bigl(\lvert\psi\rangle\!\langle\psi\rvert\bigr)E_{h_{\le R}}^\dagger,
  \label{eq:exact-deterministic-branch-channel}
\end{equation}
where
\begin{align}
  E_{h_{\le R}}={}
  &(\mathbb{I}\otimes\langle i_R\rvert\langle j_R\rvert)
  (V^{(R)}_{h_{\le R-1}}\otimes W^{(R)}_{h_{\le R-1}})
  \lvert0\rangle^{\otimes}\lvert\Phi\rangle^{\otimes}\notag\\
  &\cdots\notag\\
  &(\mathbb{I}\otimes\langle i_2\rvert\langle j_2\rvert)
  (V^{(2)}_{h_{\le1}}\otimes W^{(2)}_{h_{\le1}})
  \lvert0\rangle^{\otimes}\lvert\Phi\rangle^{\otimes}\notag\\
  &(\mathbb{I}\otimes\langle i_1\rvert\langle j_1\rvert)
  (V^{(1)}_{h_{\le0}}\otimes W^{(1)}_{h_{\le0}})
  \lvert0\rangle^{\otimes}\lvert\Phi\rangle^{\otimes}.
  \label{eq:qubit-limited-ealocc}
\end{align}
The tensor-power exponents in $\lvert0\rangle^{\otimes}$ and $\lvert\Phi\rangle^{\otimes}$ are omitted for notational simplicity.
See Figure~\ref{fig:qubit-space-overhead} for details.
\end{definition}

\begin{figure*}[t]
  \centering
  \includegraphics[width=0.95\textwidth]{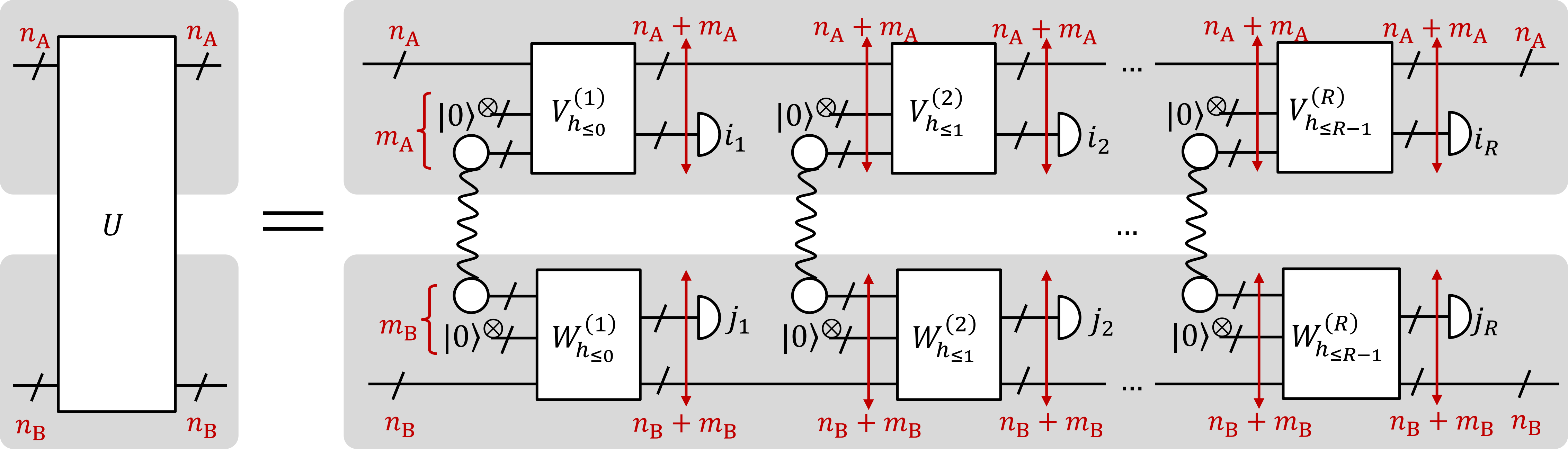}
  \caption{Definition of qubit overhead $(m_A,m_B)$.
  Red braces indicate the qubit overheads $m_A$ and $m_B$, and red arrows indicate the total qubit counts $n_A+m_A$ and $n_B+m_B$ at Alice and Bob, respectively.
  The unitary operators $V^{(r)}_{h_{\le r-1}}$ and $W^{(r)}_{h_{\le r-1}}$ depend on the sequence of measurement outcomes $h_{\le r-1}$ obtained before round $r$.
  Each measurement symbol represents computational-basis measurements on a possibly empty subset of the qubits.
  The tensor-power exponent in each $\lvert0\rangle^{\otimes}$ is omitted; the numbers of $\lvert0\rangle$ states and Bell pairs may be zero.
  At each party, the total number of auxiliary qubits initially prepared in $\lvert0\rangle$ states or as halves of Bell pairs is $m_A$ (or $m_B$).
  Before each later round, the number of qubits prepared in this way equals the number measured in the preceding round.}
  \label{fig:qubit-space-overhead}
\end{figure*}

\section{Distributed implementation of Clifford unitaries}
\label{sec:clifford}

\subsection{Operator Schmidt rank and entanglement cost}

In this subsection, we briefly review known results on the entanglement cost of bipartite unitaries when qubit overhead is unrestricted.

For a bipartite operator $U$ acting on $\cH_A\otimes\cH_B$, its \emph{operator Schmidt rank} $\OpSch(U)$ is the smallest integer $r$ for which there exist operators $A_j\in\mathcal{L}(\cH_A)$ and $B_j\in\mathcal{L}(\cH_B)$ satisfying
\begin{equation}
  U=\sum_{j=1}^{r}A_j\otimes B_j.
\end{equation}
Here, $\mathcal{L}(\cH)$ denotes the space of linear operators on $\cH$.
For a bipartite unitary $U$, the operator Schmidt rank gives a lower bound on the required entanglement resources when qubit overhead is unrestricted.
Specifically, if $U$ admits an exact deterministic LOCC implementation using a shared pure resource state of Schmidt rank $D$, then~\cite{stahlke2011}
\begin{equation}
  D\geq\OpSch(U).
\end{equation}
For a resource of $k$ Bell pairs, $D=2^k$, and hence
\begin{equation}
  k\geq \left\lceil\log_2\OpSch(U)\right\rceil.
  \label{eq:opsch-lower-bound}
\end{equation}

It is known that this bound is attainable for every bipartite Clifford unitary $U_c$~\cite{fattal2004,gottesmanchuang1999}.
In the implementation illustrated in Fig.~\ref{fig:clifford-gate-teleportation}, we first prepare the Choi state of $U_c$.
As a bipartite stabilizer state, it can always be prepared from $\log_2\OpSch(U_c)$ Bell pairs and local $\lvert0\rangle$ states using local Clifford operations~\cite{fattal2004}.
Gate teleportation using this state implements $U_c$ without additional shared entanglement~\cite{gottesmanchuang1999}.
However, holding the Choi state results in qubit overhead $(2n_A,2n_B)$.

\begin{figure}[tbp]
  \centering
  \includegraphics[width=0.9\columnwidth]{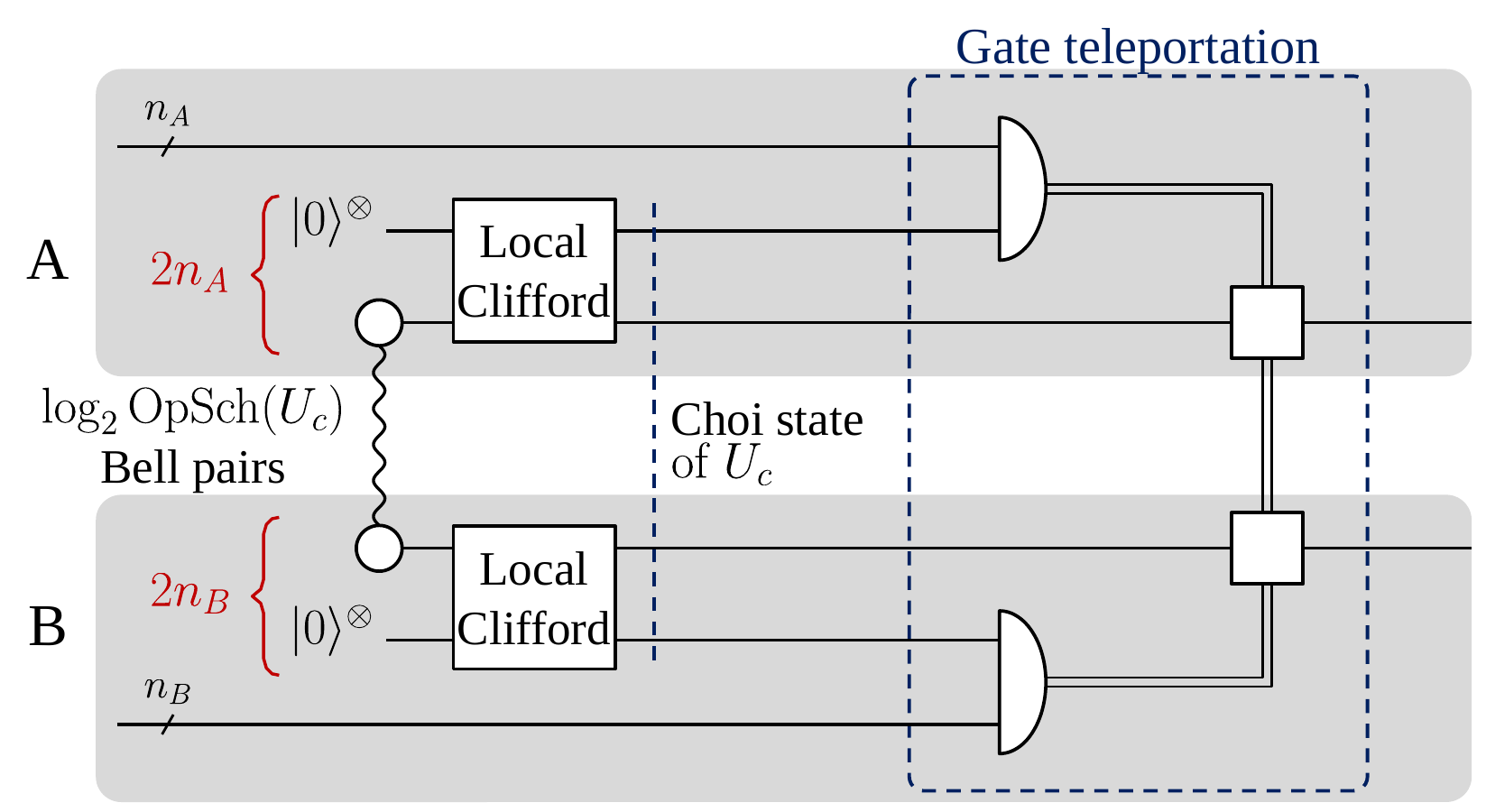}
  \caption{Implementation of a bipartite Clifford unitary $U_c$ using $\log_2\OpSch(U_c)$ Bell pairs with qubit overhead $(2n_A,2n_B)$.
  We first prepare the Choi state of $U_c$.
  For Clifford $U_c$, this can always be done by applying local Clifford operations to $\log_2\OpSch(U_c)$ Bell pairs and local $\lvert0\rangle$ states.
  The dashed box shows gate teleportation: each two-input measurement symbol represents $n_A$ or $n_B$ pairwise Bell-basis measurements on the original input and part of the Choi state, and each square represents a local Pauli correction.
  Double lines distribute both parties' measurement outcomes to both corrections.
  The two corrected registers jointly contain $U_c\lvert\psi\rangle_{AB}$ for input state $\lvert\psi\rangle_{AB}$.
  Red braces indicate the qubit overhead.}
  \label{fig:clifford-gate-teleportation}
\end{figure}

\subsection{Decomposition of Clifford unitaries with optimal entanglement consumption}
\label{sec:decomposition}

In this subsection, we introduce a decomposition lemma used to construct the distributed implementation of Clifford unitaries in Theorem~\ref{thm:main}.

Call a Clifford unitary an \emph{elementary block} if it can be obtained by applying local Clifford gates before and after one of
\begin{equation}
  \mathbb{I}_A\otimes \mathbb{I}_B,
  \qquad \CZ_{AB},
  \qquad \SWAP_{AB},
  \label{eq:elementary-blocks}
\end{equation}
where the last two operations act on one selected qubit at each party and as the identity elsewhere.
The values of $\log_2\OpSch$ for these three types are $0$, $1$, and $2$, respectively.

The distributed implementation of the identity requires neither shared entanglement nor auxiliary qubits.
A nonlocal CZ has an exact deterministic implementation using one Bell pair with qubit overhead $(1,1)$~\cite{eisert2000}, as shown in Fig.~\ref{fig:cz-implementation}.
A SWAP has an exact deterministic implementation using two Bell pairs with qubit overhead $(1,2)$, as shown in Fig.~\ref{fig:swap-implementation}.
The local Clifford gates applied before and after the operations in Eq.~\eqref{eq:elementary-blocks} require neither shared entanglement nor auxiliary qubits.
Thus every elementary block $G$ admits an exact deterministic implementation consuming $\log_2\OpSch(G)$ Bell pairs, with qubit overhead bounded by
\begin{equation}
  \begin{cases}
    (0,0),&G\sim_{\mathrm{LC}}\mathbb{I}_A\otimes\mathbb{I}_B,\\
    (1,1),&G\sim_{\mathrm{LC}}\CZ_{AB},\\
    (1,2),&G\sim_{\mathrm{LC}}\SWAP_{AB}.
  \end{cases}
  \label{eq:elementary-overhead-bounds}
\end{equation}
Here, $G\sim_{\mathrm{LC}}H$ means that $G$ can be obtained from $H$ by applying local Clifford gates before and after it.

\begin{figure}[tbp]
  \centering
  \includegraphics[height=0.356\columnwidth]{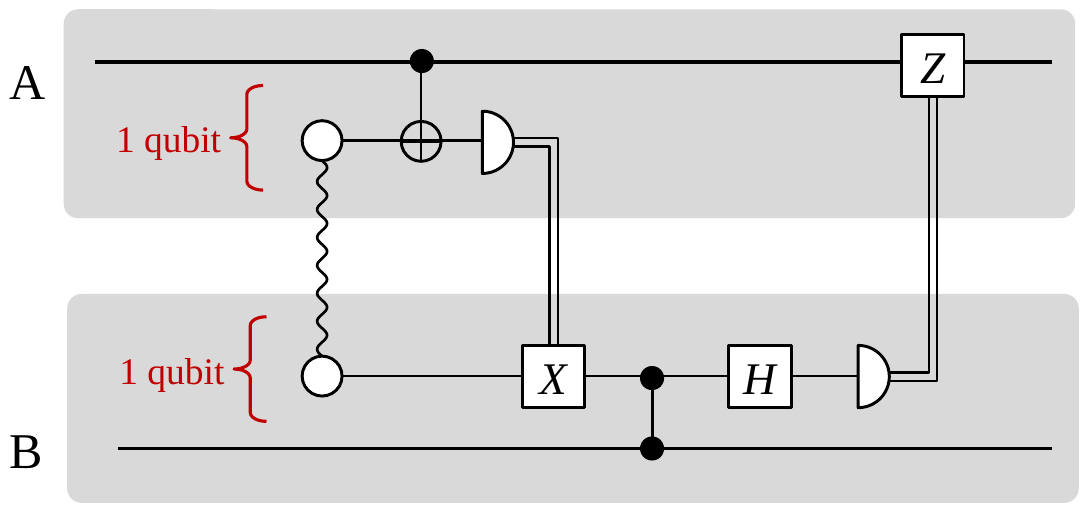}
  \caption{Distributed implementation of CZ using one Bell pair with qubit overhead $(1,1)$.
  Measurements are in the computational basis, and double lines carry their outcomes to the corresponding $X$ and $Z$ corrections.
  Red braces indicate the qubit overhead.}
  \label{fig:cz-implementation}
\end{figure}

\par
To implement a given bipartite Clifford unitary, we can decompose it into elementary blocks and implement each block separately.
However, an arbitrary decomposition of the given unitary into elementary blocks does not necessarily yield an implementation with minimum entanglement cost.
The following lemma guarantees the existence of a decomposition that achieves this minimum.

\begin{lemma}
\label{lem:unitary-decomposition}
For every bipartite Clifford unitary $U_c\in\cU(\cH_A\otimes\cH_B)$, there are elementary blocks $G_1,\ldots,G_N\in\cU(\cH_A\otimes\cH_B)$ and a phase $e^{i\theta}$ such that
\begin{equation}
  U_c=e^{i\theta}G_1G_2\cdots G_N
  \label{eq:unitary-factorization}
\end{equation}
and
\begin{equation}
  \sum_{j=1}^{N}\log_2\OpSch(G_j)
  =\log_2\OpSch(U_c).
  \label{eq:additive-cost}
\end{equation}
\end{lemma}

\noindent\textit{Proof Sketch.}
We construct the decomposition recursively, starting from $U_0:=U_c$.
Whenever $U_{j-1}$ is nonlocal, we choose a CZ-type or SWAP-type elementary block $G_j$ such that the decomposition $U_{j-1}=G_jU_j$ satisfies
\begin{equation}
  \begin{aligned}
    \log_2\OpSch(U_j)
    &=\log_2\OpSch(U_{j-1})\\
    &\quad{}-\log_2\OpSch(G_j).
  \end{aligned}
  \label{eq:sketch-cost-reduction}
\end{equation}
To find a block $G_j$ satisfying Eq.~\eqref{eq:sketch-cost-reduction}, we use the binary symplectic transformation $S_{j-1}$ corresponding to $U_{j-1}$, as described in Appendix~\ref{app:symplectic}.
In particular, we construct a symplectic transformation $R_j$ corresponding to a CZ-type or SWAP-type elementary block $G_j$ such that
\begin{equation}
  \begin{aligned}
    S_{j-1}&=R_jS_j,\\
    \rank S_j^{BA}&=\rank S_{j-1}^{BA}-\rank R_j^{BA}.
  \end{aligned}
  \label{eq:sketch-symplectic-reduction}
\end{equation}
Here, $S_j$ corresponds to $U_j$, and $\rank R_j^{BA}$ is one or two for a CZ-type or SWAP-type block, respectively.
By Lemma~\ref{lem:opsch-rank} in Appendix~\ref{app:symplectic},
\begin{equation}
  \log_2\OpSch(U_{j-1})=\rank S_{j-1}^{BA}.
  \label{eq:sketch-opsch-rank}
\end{equation}
Applying this identity also to $U_j$ and $G_j$ gives Eq.~\eqref{eq:sketch-cost-reduction} from Eq.~\eqref{eq:sketch-symplectic-reduction}.
The construction and full proof are given in Appendix~\ref{app:decomposition}.

\begin{figure}[tbp]
  \centering
  \includegraphics[height=0.35\columnwidth]{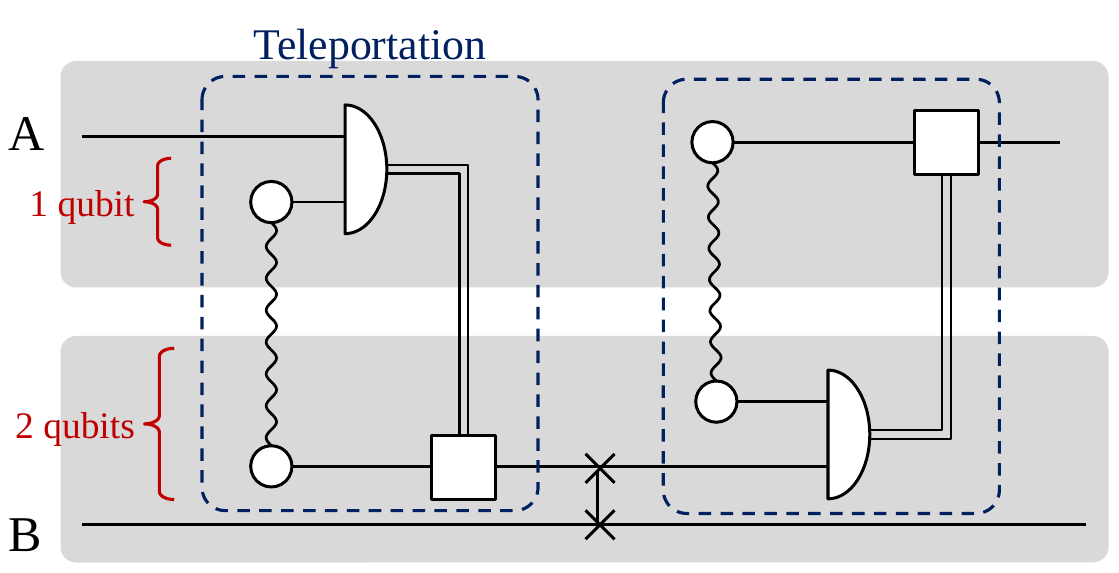}
  \caption{Distributed implementation of SWAP using two Bell pairs with qubit overhead $(1,2)$.
  The left and right dashed boxes show teleportation from Alice to Bob and from Bob to Alice, respectively.
  Each two-input measurement symbol abbreviates a Bell-basis measurement: a local CNOT gate and a Hadamard gate followed by two computational-basis measurements.
  Double lines carry the measurement outcomes to the squares representing Pauli corrections.
  Between the teleportations, Bob locally swaps the received state with his original input state, so the latter is teleported to Alice.
  Red braces indicate the qubit overhead.}
  \label{fig:swap-implementation}
\end{figure}

\subsection{Distributed implementation of Clifford unitaries with constant qubit overhead}
\label{sec:protocol}

We now combine the decomposition in Lemma~\ref{lem:unitary-decomposition} with the implementations of elementary blocks to obtain our first main result: every bipartite Clifford unitary admits a distributed implementation with minimum entanglement cost and constant qubit overhead.

\begin{theorem}
\label{thm:main}
Let $U_c\in\cU(\cH_A\otimes\cH_B)$ be a bipartite Clifford unitary.
There is an entanglement-assisted LOCC implementation of $U_c$ with qubit overhead at most $(2,2)$ that consumes exactly $\log_2\OpSch(U_c)$ Bell pairs.
\end{theorem}

\begin{proof}
Choose the decomposition $U_c=e^{i\theta}G_1\cdots G_N$ from Lemma~\ref{lem:unitary-decomposition}.
Each $G_j$ admits an exact deterministic implementation consuming $\log_2\OpSch(G_j)$ Bell pairs and using at most two auxiliary qubits at each party, by Eq.~\eqref{eq:elementary-overhead-bounds}.

Implement $G_N,\ldots,G_1$ in this order, with the rightmost factor acting first.
After implementing each elementary block, the output state is stored in $n_A$ qubits at Alice and $n_B$ qubits at Bob.
Use these output qubits as the input qubits for the next elementary block, and reset the measured qubits before implementing it.
The same auxiliary qubits can therefore be reused across elementary blocks, giving qubit overhead at most $(2,2)$.
The global phase $e^{i\theta}$ does not change the output density operator.
The resulting protocol implements $U_c$ exactly and deterministically and consumes a total of
\begin{equation}
  \sum_{j=1}^{N}\log_2\OpSch(G_j)=\log_2\OpSch(U_c)
\end{equation}
Bell pairs, where the equality follows from Lemma~\ref{lem:unitary-decomposition}.
\end{proof}

\subsection{Example and comparison with previous methods}
\label{sec:clifford-example}

We demonstrate how Theorem~\ref{thm:main} yields a distributed implementation of a Clifford unitary with optimal entanglement cost.
In this example, our construction consumes fewer Bell pairs than the two previous methods compared in Table~\ref{tab:hidden-swap-comparison}.
Let $U$ be the four-qubit Clifford unitary defined by the circuit in Fig.~\ref{fig:hidden-swap-input}, with Alice holding qubits $a_1,a_2$ and Bob holding qubits $b_1,b_2$.

\begin{figure}[tbp]
  \centering
  \includegraphics[width=0.95\columnwidth]{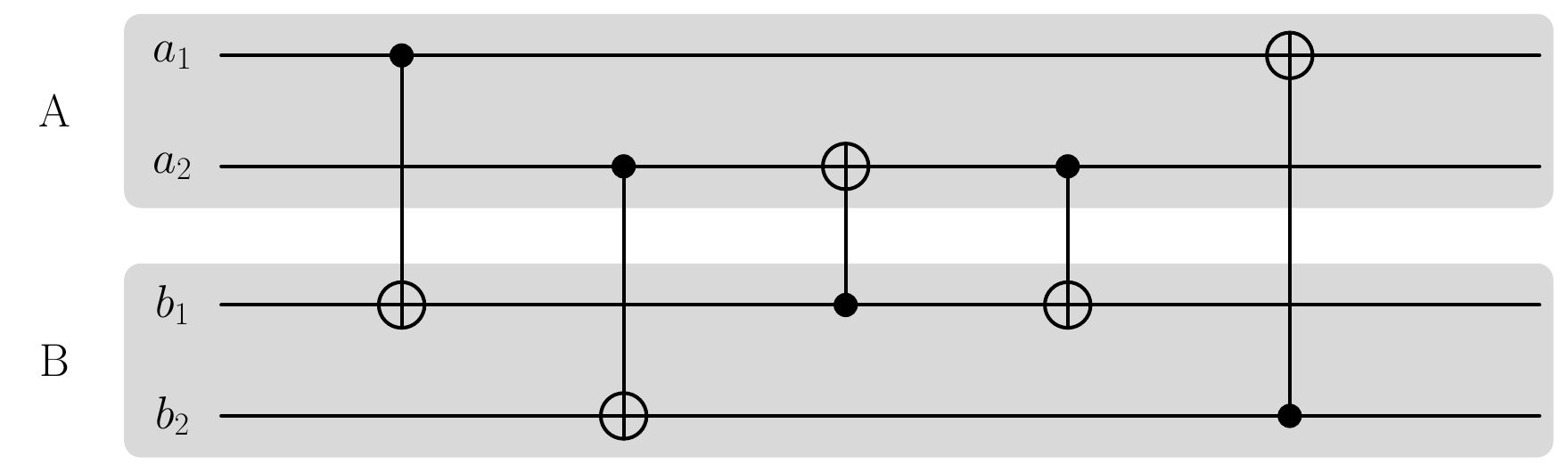}
  \caption{Quantum circuit defining the Clifford unitary $U$ used to demonstrate how Theorem~\ref{thm:main} yields a distributed implementation with optimal entanglement cost.
  Alice holds qubits $a_1,a_2$, and Bob holds qubits $b_1,b_2$.}
  \label{fig:hidden-swap-input}
\end{figure}

Applying Lemma~\ref{lem:unitary-decomposition} to $U$ yields the decomposition into three elementary blocks shown in Fig.~\ref{fig:hidden-swap-decomposition}:
\begin{equation}
  \begin{gathered}
    U=G_1G_2G_3,\\
    \left\{
    \begin{aligned}
      G_1&=D\CZ_{a_1b_2}D^\dagger,\\
      G_2&=D\SWAP_{a_2b_1}D^\dagger,\\
      G_3&=F,
    \end{aligned}
    \right.
  \end{gathered}
  \label{eq:hidden-swap-blocks}
\end{equation}
where $D$ and $F$ are the local Clifford unitaries
\begin{equation}
  \begin{aligned}
    D&=\bigl(\mathrm{CNOT}_{a_1a_2}H_{a_1}\bigr)
       \otimes\mathrm{CNOT}_{b_1b_2},\\
    F&=\bigl(\mathrm{CNOT}_{a_1a_2}\mathrm{CNOT}_{a_2a_1}\bigr)
       \otimes\bigl(\mathrm{CNOT}_{b_1b_2}\mathrm{CNOT}_{b_2b_1}\bigr).
  \end{aligned}
  \label{eq:hidden-swap-local-factors}
\end{equation}
Here, $\mathrm{CNOT}_{uv}$ denotes the controlled-NOT gate with control qubit $u$ and target qubit $v$.
The detailed derivation of this decomposition is given in Appendix~\ref{app:hidden-swap-derivation}.
Here, $G_1$ is a CZ-type elementary block, $G_2$ is a SWAP-type elementary block, and $G_3$ is an identity-type elementary block.
Following Theorem~\ref{thm:main}, we first implement $G_3$ locally, then $G_2$ using two Bell pairs, and finally $G_1$ using one Bell pair.
The resulting implementation achieves the optimal entanglement cost of three ebits, with qubit overhead at most $(2,2)$.

\begin{figure}[tbp]
  \centering
  \includegraphics[width=0.98\columnwidth]{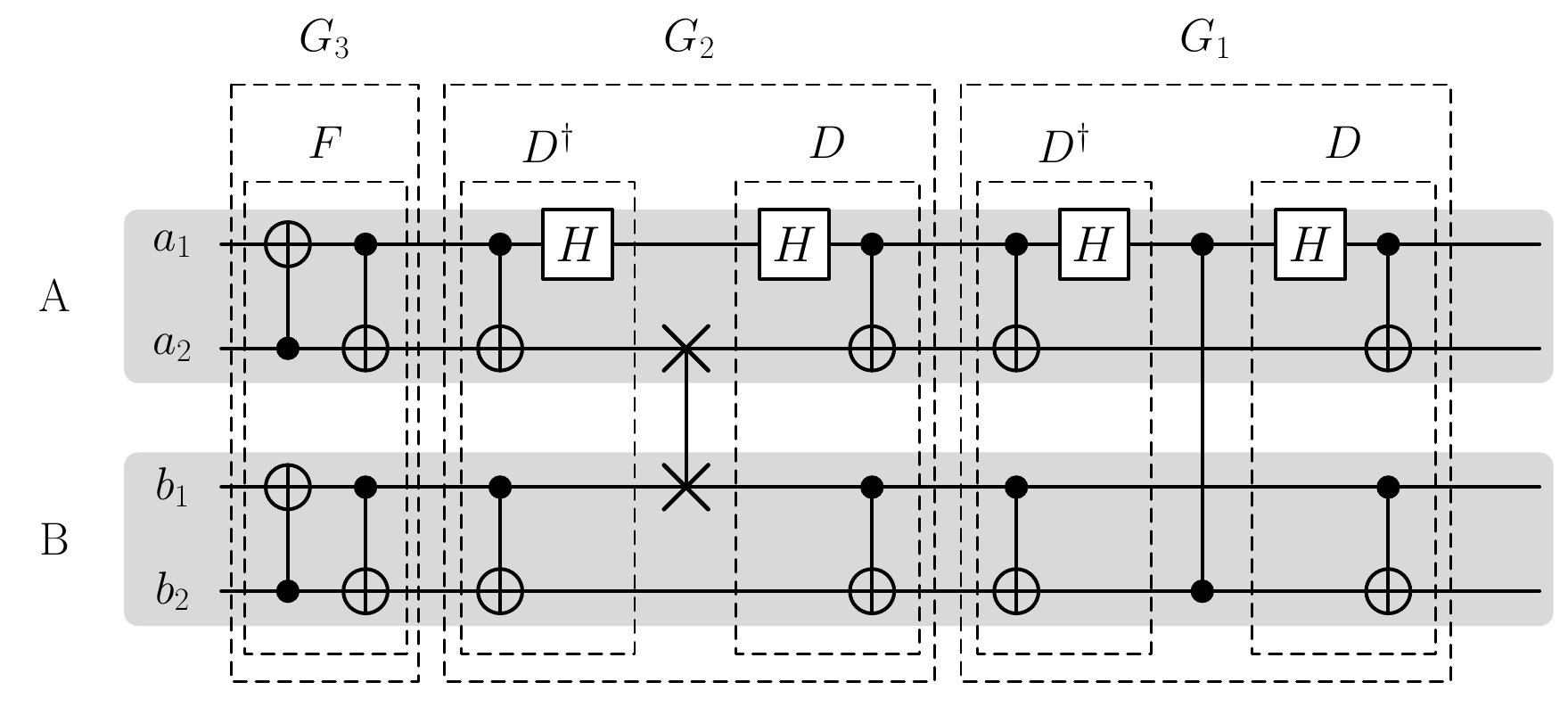}
  \caption{Decomposition of the unitary $U$ in Fig.~\ref{fig:hidden-swap-input} into three elementary blocks obtained by applying Lemma~\ref{lem:unitary-decomposition}.
  Implementing $G_3$, $G_2$, and $G_1$ in this order consumes three ebits.}
  \label{fig:hidden-swap-decomposition}
\end{figure}

Table~\ref{tab:hidden-swap-comparison} compares this entanglement cost with those obtained using two previous distributed implementation methods.
For both methods, we keep the qubit allocation fixed and convert each CNOT gate using $\mathrm{CNOT}_{uv}=H_v\CZ_{uv}H_v$.
In the hypergraph method of Andr\'{e}s-Mart\'{i}nez and Heunen~\cite{andresmartinez2019}, the intervening Hadamard gates separate the five CZ gates into groups that each consume one Bell pair.
The packing method of Wu et al.~\cite{wu2023} allows the second and fourth CZ gates to share one Bell pair by embedding the intervening third gate, giving a total cost of four Bell pairs.
The gate grouping and calculation of the Bell-pair consumption are detailed in Appendix~\ref{app:hidden-swap-derivation}.

\begin{table}[tbp]
  \caption{Entanglement costs for implementing the Clifford unitary in Fig.~\ref{fig:hidden-swap-input} with a fixed allocation of qubits $a_1,a_2$ to Alice and qubits $b_1,b_2$ to Bob.
  To calculate the Bell-pair consumption for the first two methods, we convert each CNOT gate in the original circuit using $\mathrm{CNOT}_{uv}=H_v\CZ_{uv}H_v$.
  See Appendix~\ref{app:hidden-swap-derivation} for details.}
  \label{tab:hidden-swap-comparison}
  \begin{ruledtabular}
    \begin{tabular}{lc}
      \shortstack[l]{Distributed implementation\\method} & \shortstack{Number of Bell\\pairs consumed} \\
      \hline
      Andr\'{e}s-Mart\'{i}nez and Heunen~\cite{andresmartinez2019} & 5 \\
      Wu et al.~\cite{wu2023} & 4 \\
      This work (Theorem~\ref{thm:main}) & 3 (optimal)
    \end{tabular}
  \end{ruledtabular}
\end{table}

\section{Distributed implementation of unitaries with Clifford+\texorpdfstring{$T$}{T} decompositions}
\label{sec:clifford-t}

We now consider unitaries specified by exact Clifford+$T$ decompositions, where the $T$ gate is defined by $T:=\operatorname{diag}(1,e^{i\pi/4})$ in the computational basis.
Building on Theorem~\ref{thm:main}, we construct entanglement-assisted LOCC implementations of these unitaries and derive an upper bound on their entanglement cost.
These implementations retain qubit overhead at most $(2,2)$, and their Bell-pair cost is bounded in terms of the operator Schmidt rank of the target unitary and the $T$-count of the given decomposition.

\begin{theorem}
\label{thm:clifford-t-constant-space}
Suppose that $U\in\cU(\cH_A\otimes\cH_B)$ has the following Clifford+$T$ decomposition containing $t\in\mathbb{Z}_{\geq0}$ $T$ gates:
\begin{equation}
  U=e^{i\phi}C_tT_{q_t}C_{t-1}T_{q_{t-1}}\cdots C_1T_{q_1}C_0,
  \label{eq:exact-Clifford-T-decomposition}
\end{equation}
where $\phi\in\mathbb{R}$.
For $j=0,\ldots,t$, each $C_j\in\cU(\cH_A\otimes\cH_B)$ is a Clifford unitary.
For $j=1,\ldots,t$, the qubit $q_j$ belongs to Alice or Bob, and $T_{q_j}$ acts as $T$ on that qubit and as the identity elsewhere.
There is an entanglement-assisted LOCC implementation of $U$ with qubit overhead at most $(2,2)$ that consumes $K\in\mathbb{Z}_{\geq0}$ Bell pairs satisfying
\begin{equation}
  K\leq\left\lfloor\log_2\OpSch(U)\right\rfloor+2t.
  \label{eq:clifford-t-constant-space-cost}
\end{equation}
\end{theorem}

\noindent\textit{Proof Sketch.}
For a signed Pauli operator $P\in\cU(\cH_A\otimes\cH_B)$, meaning a tensor product of $\mathbb{I},X,Y,Z$ multiplied by $+1$ or $-1$, write $R(P):=e^{-i\pi P/8}$.
Commuting all Clifford factors to the right gives
\begin{equation}
  U=e^{i\phi'}R(P_t)\cdots R(P_1)C,
  \label{eq:sketch-Pauli-rotation-form}
\end{equation}
where $C:=C_t\cdots C_0$ is a Clifford unitary and $\phi':=\phi+t\pi/8$.
For each $j$, $P_j$ is obtained by conjugating $Z_{q_j}$ by $C_t\cdots C_j$.
Each $R(P_j)$ is either local or locally Clifford-equivalent to a unitary controlled by one qubit, and hence can be implemented using at most one Bell pair with qubit overhead at most $(1,1)$.
Implementing $C$ using Theorem~\ref{thm:main} and then $R(P_1),\ldots,R(P_t)$ in this order gives the Bell-pair cost bound
\begin{equation}
  K\leq\log_2\OpSch(C)+t,
  \label{eq:sketch-rotation-cost}
\end{equation}
with qubit overhead at most $(2,2)$.
Since $C=e^{-i\phi'}R(P_1)^\dagger\cdots R(P_t)^\dagger U$, submultiplicativity of the operator Schmidt rank gives
\begin{equation}
  \OpSch(C)\leq2^t\OpSch(U).
  \label{eq:sketch-Clifford-rank-bound}
\end{equation}
Combining this bound with Eq.~\eqref{eq:sketch-rotation-cost} yields Eq.~\eqref{eq:clifford-t-constant-space-cost}.
The full proof is given in Appendix~\ref{app:clifford-t}.

\section{Discussion and conclusions}
\label{sec:conclusions}

We have studied the entanglement cost of exact deterministic implementations of bipartite unitaries with constant qubit overhead.
For every bipartite Clifford unitary, the operator Schmidt rank lower bound can be attained using at most two auxiliary qubits at each QPU.
Thus, restricting the qubit overhead to a constant does not increase the minimum entanglement cost for bipartite Clifford unitaries.
For unitaries specified by an exact Clifford+$T$ decomposition with $T$-count $t$, we have also given an implementation using at most two auxiliary qubits at each QPU, with an entanglement cost at most $2t$ Bell pairs above the operator Schmidt rank lower bound.

For non-Clifford unitaries, our upper bound on entanglement cost is not tight in general.
An important open problem is to determine the minimum entanglement cost of non-Clifford unitaries under qubit-space constraints.
However, even without qubit-space constraints, the minimum entanglement cost of non-Clifford unitaries has not been characterized in general.
Beyond determining this cost, the space-constrained setting therefore raises the additional question of whether, and by how much, qubit-space constraints increase the minimum entanglement cost of a non-Clifford unitary.

\begin{acknowledgments}
This work was supported by the Japan Society for the Promotion of Science (JSPS) KAKENHI Grant Number \mbox{23K21643}, the MEXT Quantum Leap Flagship Program (MEXT QLEAP) \mbox{JPMXS0118069605} and \mbox{JPMXS0120351339}, JST CREST Grant Number \mbox{JPMJCR25I5}, JST ASPIRE Grant Number \mbox{JPMJAP25A3}, JST NEXUS Grant Number \mbox{JPMJNX26C9}, FoPM, WINGS Program, the University of Tokyo, and IBM Quantum.
\end{acknowledgments}

\paragraph*{Data Availability.}
No new data were created or analyzed in this study.

\appendix

\section{Symplectic representation of Clifford unitaries}
\label{app:symplectic}

This appendix reviews the basics of the symplectic representation of Clifford unitaries and establishes additional lemmas for our analysis of bipartite Clifford unitaries.

\subsection{Preliminaries on symplectic spaces}

We first collect the required definitions and properties of symplectic spaces.
All vector spaces considered here are finite-dimensional and over the field $\F_2$ with two elements.

Let $V$ be a vector space equipped with a bilinear form $[\cdot,\cdot]:V\times V\to\F_2$.
We say that the form is \emph{nondegenerate} if
\begin{equation}
  \forall v\in V,\quad
  \bigl(\forall w\in V,\ [v,w]=0\bigr)\Longrightarrow v=0.
  \label{eq:nondegenerate-form}
\end{equation}
We also say that the space $V$ is \emph{nondegenerate} when its given bilinear form is nondegenerate.

A \emph{symplectic form} on a vector space $V$ is a nondegenerate alternating bilinear form.
Alternation means $[v,v]=0$ for every $v\in V$.
A vector space $V$ equipped with a symplectic form is called a \emph{symplectic space}.
A \emph{symplectic transformation} of a symplectic space $(V,[\cdot,\cdot])$ is an invertible linear map $S:V\to V$ that preserves the form, namely $[Sv,Sw]=[v,w]$ for all $v,w\in V$.
The group of these transformations is denoted by $\Sp(V)$.

Let $W$ be a subspace of a symplectic space $V$ with form $[\cdot,\cdot]$.
On $W$, we use the restriction of this form to $W\times W$.
Define the \emph{orthogonal complement} of $W$ in $V$ and the \emph{radical} of $W$ by
\begin{equation}
  W^\perp
  :=\{v\in V:[v,w]=0\ \text{for all }w\in W\},
\end{equation}
\begin{equation}
  \begin{aligned}
    \rad(W)&:=W\cap W^\perp\\
    &=\{w\in W:[w,x]=0\ \text{for all }x\in W\}.
  \end{aligned}
  \label{eq:radical-definition}
\end{equation}
Equation~\eqref{eq:radical-definition} shows that $W$ is nondegenerate if and only if $\rad(W)=\{0\}$.
Equivalently,
\begin{equation}
  W\text{ is nondegenerate}
  \quad\Longleftrightarrow\quad
  W\cap W^\perp=\{0\}.
  \label{eq:nondegenerate-subspace}
\end{equation}

\begin{lemma}
\label{lem:orthogonal-complement}
Let $V$ be a finite-dimensional symplectic space over $\F_2$, and let $W\subseteq V$ be any linear subspace.
Then
\begin{equation}
  \dim W+\dim W^\perp=\dim V,
  \label{eq:orthogonal-complement-dimension}
\end{equation}
and
\begin{equation}
  (W^\perp)^\perp=W.
  \label{eq:double-perp}
\end{equation}
If $W$ is nondegenerate, then
\begin{equation}
  V=W\oplus W^\perp,
  \label{eq:orthogonal-direct-sum}
\end{equation}
and $W^\perp$ is nondegenerate.
\end{lemma}

\paragraph*{Proof.}
We first note that the form is symmetric over $\F_2$.
Bilinearity and alternation give
\begin{equation}
  0=[x+y,x+y]=[x,y]+[y,x]
  \qquad(x,y\in V).
\end{equation}
Since the field is $\F_2$, we have $[x,y]=[y,x]$.

\paragraph*{Proof of Eq.~\eqref{eq:orthogonal-complement-dimension}.}
Let $V^*$ and $W^*$ denote the dual spaces of $V$ and $W$, respectively.
Define the linear map $\Phi:V\to W^*$ by $\Phi(v)(w):=[w,v]$.
The rank--nullity theorem gives $\dim V=\dim\ker\Phi+\rank\Phi$.
Thus it suffices to show that $\ker\Phi=W^\perp$ and $\rank\Phi=\dim W$.

We first show that $\ker\Phi=W^\perp$.
For $v\in V$, the functional $\Phi(v)$ is zero if and only if $[w,v]=0$ for every $w\in W$.
By symmetry and the definition of $W^\perp$, this is equivalent to $v\in W^\perp$.

We next show that $\rank\Phi=\dim W$ by proving that $\Phi$ is surjective onto $W^*$.
For an arbitrary $g\in W^*$, we will find $v\in V$ such that $\Phi(v)=g$.

First, extend the linear functional $g:W\to\mathbb{F}_2$ to a linear functional $\widetilde g:V\to\mathbb{F}_2$ that agrees with $g$ on $W$.
To construct this extension, write $k=\dim W$ and $d=\dim V$.
Choose a basis $w_1,\ldots,w_k$ of $W$ and extend it to a basis $w_1,\ldots,w_k,u_{k+1},\ldots,u_d$ of $V$.
The assignments $\widetilde g(w_i):=g(w_i)$ and $\widetilde g(u_j):=0$ determine a linear functional $\widetilde g\in V^*$ satisfying $\widetilde g(w)=g(w)$ for every $w\in W$.

Next, we represent $\widetilde g$ in the form $\widetilde g(x)=[x,v]$ for some $v\in V$.
For this purpose, define the linear map $J:V\to V^*$ by $J(v)(x):=[x,v]$.
If $J(v)$ is the zero functional, then $[x,v]=0$ for every $x\in V$.
By symmetry and nondegeneracy, this implies $v=0$.
Thus $J$ is injective, and since $\dim V=\dim V^*$, it is an isomorphism.

We can therefore choose $v\in V$ such that $J(v)=\widetilde g$.
For every $w\in W$,
\begin{equation}
  \Phi(v)(w)=[w,v]=\widetilde g(w)=g(w).
\end{equation}
Thus $\Phi(v)=g$, and since $g$ was arbitrary, $\Phi$ is surjective.
Consequently, $\rank\Phi=\dim W^*=\dim W$.

Substituting the kernel and rank into the rank--nullity theorem gives
\begin{equation}
  \dim V=\dim\ker\Phi+\rank\Phi
  =\dim W^\perp+\dim W.
\end{equation}

\paragraph*{Proof of Eq.~\eqref{eq:double-perp}.}
We prove the equality by showing that $W\subseteq(W^\perp)^\perp$ and that the two spaces have the same dimension.
For $w\in W$ and $u\in W^\perp$, symmetry gives $[w,u]=[u,w]=0$.
Thus $W\subseteq(W^\perp)^\perp$.
Applying Eq.~\eqref{eq:orthogonal-complement-dimension} to the subspace $W^\perp$ yields
\begin{equation}
  \dim(W^\perp)^\perp=\dim V-\dim W^\perp=\dim W.
\end{equation}
The inclusion and equality of dimensions imply $(W^\perp)^\perp=W$.

\paragraph*{Proof of the nondegenerate case.}
Suppose that $W$ is nondegenerate.
To prove the direct-sum decomposition in Eq.~\eqref{eq:orthogonal-direct-sum}, we show that $W$ and $W^\perp$ have zero intersection and together span $V$.
Nondegeneracy gives $W\cap W^\perp=\{0\}$ by Eq.~\eqref{eq:nondegenerate-subspace}.
Together with Eq.~\eqref{eq:orthogonal-complement-dimension}, this yields
\begin{equation}
  \dim(W+W^\perp)=\dim W+\dim W^\perp=\dim V.
\end{equation}
Since $W+W^\perp\subseteq V$, this gives $W+W^\perp=V$.
The zero intersection makes the sum direct, so $V=W\oplus W^\perp$.

It remains to show that $W^\perp$ is nondegenerate, which means that its radical is zero.
Using Eq.~\eqref{eq:double-perp}, we obtain
\begin{equation}
  \rad(W^\perp)
  =W^\perp\cap(W^\perp)^\perp
  =W^\perp\cap W
  =\{0\},
\end{equation}
which proves that $W^\perp$ is nondegenerate.  \hfill$\square$

Vectors $u,v\in V$ form a \emph{symplectic pair} if $[u,v]=1$.
A basis $(u_1,v_1,\ldots,u_m,v_m)$ of $V$ is \emph{symplectic} if
\begin{equation}
  [u_j,u_k]=[v_j,v_k]=0,\qquad [u_j,v_k]=\delta_{jk}
\end{equation}
for all $1\leq j,k\leq m$.

\begin{lemma}
\label{lem:symplectic-basis-extension}
Let $V$ be a finite-dimensional symplectic space over $\F_2$, and let $(u,v)$ be a symplectic pair in $V$.
Then $V$ has a symplectic basis containing $u$ and $v$.
\end{lemma}

\paragraph*{Proof.}
Set $H=\operatorname{span}\{u,v\}$.
Since $[u,v]=1$, the vectors $u$ and $v$ are linearly independent.
We first show that $H$ is nondegenerate.
For $h=\alpha u+\beta v\in H$ with $\alpha,\beta\in\F_2$, we have $[h,v]=\alpha$ and $[h,u]=\beta$.
Suppose that $h$ is orthogonal to every vector in $H$.
Since $u,v\in H$, we have $[h,u]=[h,v]=0$.
The two identities therefore give $\alpha=\beta=0$, and hence $h=0$.
Thus $H$ is nondegenerate.

Lemma~\ref{lem:orthogonal-complement} gives $V=H\oplus H^\perp$, with $H^\perp$ nondegenerate.
Set $(u_1,v_1)=(u,v)$ and $V_1=H^\perp$.
For each $j\geq1$ with $V_j\neq\{0\}$, choose $0\neq u_{j+1}\in V_j$.
Nondegeneracy of $V_j$ gives $v_{j+1}\in V_j$ with $[u_{j+1},v_{j+1}]=1$.
The subspace $\operatorname{span}\{u_{j+1},v_{j+1}\}$ is nondegenerate by the argument above.
Lemma~\ref{lem:orthogonal-complement} gives the decomposition
\begin{equation}
  V_j=\operatorname{span}\{u_{j+1},v_{j+1}\}\oplus V_{j+1},
\end{equation}
where $V_{j+1}$ is the nondegenerate orthogonal complement of $\operatorname{span}\{u_{j+1},v_{j+1}\}$ within $V_j$.
The construction can continue whenever $V_j\neq\{0\}$.
Since $\dim V_{j+1}<\dim V_j$, the process terminates after finitely many steps with $V_m=\{0\}$ for some $m$.

The resulting orthogonal direct-sum decomposition is
\begin{equation}
  V=\bigoplus_{j=1}^{m}\operatorname{span}\{u_j,v_j\}.
\end{equation}
Thus $(u_1,v_1,\ldots,u_m,v_m)$ is a basis of $V$.
Each pair satisfies $[u_j,v_j]=1$, and distinct pairs are mutually orthogonal by construction.
Hence this basis is symplectic and contains $u$ and $v$.  \hfill$\square$

\subsection{Lemmas on bipartite Clifford unitaries}
\label{app:pauli-clifford}

We now use this framework to represent Pauli operators and Clifford unitaries.

For each party $X\in\{A,B\}$, we represent Pauli operators up to phase by vectors in $V_X=\F_2^{2n_X}$.
For $p=(x,z)$ with $x,z\in\F_2^n$, write $P(p)$ for the $n$-qubit Pauli operator proportional to $\bigotimes_{j=1}^n X_j^{x_j}Z_j^{z_j}$.
The commutation relations of Pauli operators are determined by the symplectic form
\begin{equation}
  [p,q]:=x\mathbin{\cdot}z'+z\mathbin{\cdot}x'
  \pmod{2},
  \label{eq:symplectic-form}
\end{equation}
where $q=(x',z')$.
The operators $P(p)$ and $P(q)$ commute if $[p,q]=0$ and anticommute if $[p,q]=1$.
This bilinear form is alternating and nondegenerate, so each $V_X$ is a symplectic space.
We equip $V=V_A\oplus V_B$ with the orthogonal direct-sum form, making $V$ a symplectic space as well.
We write $\perp_A$ and $\perp_B$ for orthogonal complements taken within $V_A$ and $V_B$, respectively.

A Clifford unitary $U_c$ maps Pauli operators to Pauli operators under conjugation.
In the binary representation, this action is described by an invertible linear transformation $S:V\to V$ satisfying
\begin{equation}
  U_c P(p) U_c^\dagger\propto P(Sp),\qquad \forall p\in V.
\end{equation}
Since conjugation preserves commutation relations, $S$ preserves the symplectic form and is therefore a symplectic transformation, i.e., $S\in\Sp(V)$.
Relative to $V_A\oplus V_B$, write
\begin{equation}
  S=
  \begin{pmatrix}
    S^{AA}&S^{AB}\\
    S^{BA}&S^{BB}
  \end{pmatrix}.
  \label{eq:block-S}
\end{equation}
Here, $S^{YX}:V_X\to V_Y$ for $X,Y\in\{A,B\}$, and all matrix ranks below are over $\F_2$.

\begin{lemma}
\label{lem:block-diagonal-symplectic}
Let $U_c\in\cU(\cH_A\otimes\cH_B)$ be a Clifford unitary with symplectic transformation $S\in\Sp(V_A\oplus V_B)$.
If $S^{BA}=0$, then $U_c=U_A\otimes U_B$ for Clifford unitaries $U_A\in\cU(\cH_A)$ and $U_B\in\cU(\cH_B)$.
\end{lemma}

\paragraph*{Proof.}
We first show that $S$ is block diagonal.
The condition $S^{BA}=0$ gives $S(V_A)\subseteq V_A$.
Since $S$ is invertible, $\dim S(V_A)=\dim V_A$, so $S(V_A)=V_A$.

To determine $S(V_B)$, note that a symplectic transformation satisfies $S(W^\perp)=S(W)^\perp$ for every subspace $W\subseteq V_A\oplus V_B$.
Indeed, form preservation gives $[Sx,Sw]=[x,w]$ for all $x\in V_A\oplus V_B$ and $w\in W$.
Thus $x\in W^\perp$ if and only if $Sx\in S(W)^\perp$, and since $S$ is bijective, $S(W^\perp)=S(W)^\perp$.
The spaces $V_A$ and $V_B$ are nondegenerate and orthogonal under the symplectic form on $V_A\oplus V_B$, so $V_B=V_A^\perp$.
Applying $S(W^\perp)=S(W)^\perp$ with $W=V_A$ and using $S(V_A)=V_A$, we obtain
\begin{equation}
  S(V_B)=S(V_A^\perp)=S(V_A)^\perp=V_A^\perp=V_B,
\end{equation}
so $S^{AB}=0$ and $S$ is block diagonal.

The blocks $S^{AA}$ and $S^{BB}$ are the restrictions of $S$ to $V_A$ and $V_B$, respectively.
Since $S$ is invertible and satisfies $S(V_A)=V_A$ and $S(V_B)=V_B$, $S^{AA}:V_A\to V_A$ and $S^{BB}:V_B\to V_B$ are invertible.
For $x,y\in V_A$, preservation of the symplectic form by $S$ gives
\begin{equation}
  [S^{AA}x,S^{AA}y]=[Sx,Sy]=[x,y].
\end{equation}
Hence $S^{AA}$ preserves the form on $V_A$, and the same calculation for $x,y\in V_B$ shows that $S^{BB}$ preserves the form on $V_B$.
Therefore $S^{AA}\in\Sp(V_A)$ and $S^{BB}\in\Sp(V_B)$.

Choose Clifford unitaries $\widetilde U_A\in\cU(\cH_A)$ and $\widetilde U_B\in\cU(\cH_B)$ representing $S^{AA}$ and $S^{BB}$, respectively.
The tensor product $\widetilde U_A\otimes\widetilde U_B$ has the same symplectic transformation $S$ as $U_c$.
Therefore $U_c$ differs from $\widetilde U_A\otimes\widetilde U_B$ by a Pauli operator and a global phase.
Every bipartite Pauli is a tensor product of Pauli operators on $\cH_A$ and $\cH_B$, so
\begin{equation}
  U_c=(e^{i\theta}P_A\widetilde U_A)\otimes(P_B\widetilde U_B)
\end{equation}
for Pauli operators $P_A\in\cU(\cH_A)$ and $P_B\in\cU(\cH_B)$ and some $\theta\in\mathbb{R}$.
Thus $U_A:=e^{i\theta}P_A\widetilde U_A$ and $U_B:=P_B\widetilde U_B$ are Clifford unitaries satisfying $U_c=U_A\otimes U_B$.  \hfill$\square$

The following lemma expresses the operator Schmidt rank in terms of the symplectic representation.

\begin{lemma}
\label{lem:opsch-rank}
For a bipartite Clifford unitary $U_c\in\cU(\cH_A\otimes\cH_B)$ with symplectic transformation $S\in\Sp(V_A\oplus V_B)$,
\begin{equation}
  \log_2\OpSch(U_c)=\rank S^{BA}.
  \label{eq:opsch-cross-rank}
\end{equation}
\end{lemma}

\paragraph*{Proof.}
To relate the operator Schmidt rank of $U_c$ to $\rank S^{BA}$, we use its Choi state.
Introduce reference registers $A'$ and $B'$ with Hilbert spaces $\cH_{A'}\cong\cH_A$ and $\cH_{B'}\cong\cH_B$, containing $n_A$ and $n_B$ qubits, respectively.
The operator Schmidt rank of $U_c$ equals the Schmidt rank, across $AA'|BB'$, of its normalized Choi state
\begin{equation}
  \lvert U_c\rangle
  =(U_c\otimes\mathbb{I}_{A'B'})
  \lvert\Phi\rangle_{AA'}^{\otimes n_A}
  \lvert\Phi\rangle_{BB'}^{\otimes n_B}.
  \label{eq:choi-vector}
\end{equation}
This is a stabilizer state, so its entanglement can be found by counting the stabilizer elements supported only on $AA'$.
To identify these elements, we may omit overall phases and write the stabilizer elements as
\begin{equation}
  P\bigl(S(v,w)\bigr)_{AB}
  \otimes P(v)_{A'}^{\mathsf T}
  \otimes P(w)_{B'}^{\mathsf T},
  \label{eq:choi-stabilizers}
\end{equation}
where the elements are parametrized by $(v,w)\in V_A\oplus V_B$.
The superscript $\mathsf T$ denotes transpose in the computational basis.
Such an element is supported only on $AA'$ if and only if it acts as the identity on both $B$ and $B'$.
The identity action on $B'$ forces $w=0$, and the identity action on $B$ then requires $S^{BA}v=0$.
The binary labels of the stabilizer elements supported only on $AA'$ therefore form a vector space of dimension
\begin{equation}
  \dim\ker S^{BA}=2n_A-\rank S^{BA},
\end{equation}
where the equality follows from the rank--nullity theorem applied to $S^{BA}:V_A\to V_B$, using $\dim V_A=2n_A$.

For a bipartite pure stabilizer state on $N_A$ qubits at Alice, the entanglement entropy in bits is $N_A$ minus the number of independent stabilizer generators supported only on Alice's subsystem~\cite{fattal2004}.
Since $AA'$ consists of $2n_A$ qubits, the entanglement entropy of $\lvert U_c\rangle$ across $AA'|BB'$ is
\begin{equation}
  2n_A-\dim\ker S^{BA}
  =\rank S^{BA}.
\end{equation}
The nonzero Schmidt coefficients of $\lvert U_c\rangle$ are all equal, so this entropy also equals the base-two logarithm of its Schmidt rank, namely $\log_2\OpSch(U_c)$.
This proves the lemma.  \hfill$\square$

\section{Proof of Lemma~\ref{lem:unitary-decomposition}}
\label{app:decomposition}

We first describe the symplectic transformations associated with the CZ-type and SWAP-type elementary blocks.
Recall from Appendix~\ref{app:pauli-clifford} that $V_A=\F_2^{2n_A}$ and $V_B=\F_2^{2n_B}$ are the symplectic spaces used to represent, up to phase, Pauli operators on Alice's and Bob's systems, respectively.

\begin{lemma}[Symplectic transformation of a CZ-type elementary block]
\label{lem:cz-symplectic}
Let $q\in V_A\setminus\{0\}$ and $b\in V_B\setminus\{0\}$.
Define a linear map $R:V_A\oplus V_B\to V_A\oplus V_B$ by
\begin{equation}
  \begin{aligned}
    R(x)&=x+[x,q]b &&(x\in V_A),\\
    R(y)&=y+[y,b]q &&(y\in V_B).
  \end{aligned}
  \label{eq:cz-linear-action}
\end{equation}
Then $R$ is symplectic and represents a CZ-type elementary block, with $R^{-1}=R$ and $\rank R^{BA}=1$.
Moreover, $R$ fixes every $x\in V_A$ with $[x,q]=0$ and every $y\in V_B$ with $[y,b]=0$.
\end{lemma}

\begin{proof}
Choose $a\in V_A$ and $s\in V_B$ with $[q,a]=[b,s]=1$, which exist by nondegeneracy of $V_A,V_B$ and the assumptions $q,b\neq0$.
By Lemma~\ref{lem:symplectic-basis-extension}, extend the symplectic pairs $(q,a)$ and $(b,s)$ to local symplectic bases.
Equation~\eqref{eq:cz-linear-action} gives
\begin{equation}
  q\mapsto q,\quad a\mapsto a+b,\quad
  b\mapsto b,\quad s\mapsto s+q.
  \label{eq:cz-basis-action}
\end{equation}
Every other vector in Alice's basis is orthogonal to $q$, and every other vector in Bob's basis is orthogonal to $b$, so $R$ fixes these vectors by Eq.~\eqref{eq:cz-linear-action}.
Thus, for suitable local symplectic transformations $D_A\in\Sp(V_A)$ and $D_B\in\Sp(V_B)$,
\begin{equation}
  R=(D_A\oplus D_B)C_{\CZ}(D_A\oplus D_B)^{-1}.
\end{equation}
Here, $D_A,D_B$ map the standard local symplectic bases to the chosen bases, sending the selected qubits' $Z,X$ labels to $(q,a)$ and $(b,s)$, respectively.
The map $C_{\CZ}$ is the symplectic matrix of CZ in the standard bases.
Since $D_A,D_B$ have local Clifford representatives, $R$ is symplectic and represents a CZ-type elementary block.
The identity $C_{\CZ}^2=\mathbb{I}$ also gives $R^{-1}=R$.

To compute the rank, we use Eq.~\eqref{eq:cz-linear-action}, which gives $R^{BA}x=[x,q]b$ for $x\in V_A$.
Since $R^{BA}a=b\neq0$, we have $\operatorname{im}R^{BA}=\operatorname{span}\{b\}$ and $\rank R^{BA}=1$.
The assertions about fixed vectors follow directly from Eq.~\eqref{eq:cz-linear-action}.
\end{proof}

\begin{lemma}[Symplectic transformation of a SWAP-type elementary block]
\label{lem:swap-symplectic}
Let $(p_1,p_2)$ and $(h_1,h_2)$ be symplectic pairs in $V_A$ and $V_B$, respectively, and set $V_A'=\operatorname{span}\{p_1,p_2\}$ and $V_B'=\operatorname{span}\{h_1,h_2\}$.
Define a linear map $R:V_A\oplus V_B\to V_A\oplus V_B$ by
\begin{equation}
  \begin{aligned}
    R(p_j)&=h_j,\quad R(h_j)=p_j\quad(j=1,2),\\
    R(z)&=z\quad\bigl(z\in (V_A')^{\perp_A}\oplus (V_B')^{\perp_B}\bigr).
  \end{aligned}
  \label{eq:swap-pair-action}
\end{equation}
Then $R$ is symplectic and represents a SWAP-type elementary block, with $R^{-1}=R$ and $\rank R^{BA}=2$.
\end{lemma}

\begin{proof}
The spaces $V_A'$ and $V_B'$ are nondegenerate, as can be checked directly from the definition in Eq.~\eqref{eq:nondegenerate-form} using the symplectic-pair relations $[p_1,p_2]=[h_1,h_2]=1$.
Lemma~\ref{lem:orthogonal-complement} gives $V_A=V_A'\oplus (V_A')^{\perp_A}$ and $V_B=V_B'\oplus (V_B')^{\perp_B}$, so Eq.~\eqref{eq:swap-pair-action} defines a linear map on all of $V_A\oplus V_B$.
Use Lemma~\ref{lem:symplectic-basis-extension} to extend the symplectic pairs $(p_1,p_2)$ and $(h_1,h_2)$ to local symplectic bases.
In these bases, $R$ exchanges the pairs $(p_1,p_2)$ and $(h_1,h_2)$ and fixes every other basis vector.
Thus, for suitable local symplectic transformations $D_A\in\Sp(V_A)$ and $D_B\in\Sp(V_B)$,
\begin{equation}
  R=(D_A\oplus D_B)C_{\SWAP}(D_A\oplus D_B)^{-1}.
\end{equation}
Here, $D_A,D_B$ map the standard local symplectic bases to the chosen bases, sending the selected qubits' $Z,X$ labels to $(p_1,p_2)$ and $(h_1,h_2)$, respectively.
The map $C_{\SWAP}$ is the symplectic matrix of SWAP in the standard bases.
Since $D_A,D_B$ have local Clifford representatives, $R$ is symplectic and represents a SWAP-type elementary block.
The identity $C_{\SWAP}^2=\mathbb{I}$ also gives $R^{-1}=R$.

To compute the rank, we use the definition of $R$, which gives $R^{BA}p_j=h_j$ for $j=1,2$ and $R^{BA}u=0$ for $u\in (V_A')^{\perp_A}$.
Using $V_A=V_A'\oplus (V_A')^{\perp_A}$, we obtain $\operatorname{im}R^{BA}=V_B'$ and $\rank R^{BA}=2$.
\end{proof}

To prove Lemma~\ref{lem:unitary-decomposition}, we first decompose the symplectic transformation of $U_c$ into a product of transformations representing elementary blocks, while satisfying a condition on the $BA$-block ranks.
We then use the corresponding Clifford unitaries to obtain a decomposition of $U_c$.

\begin{lemma}
\label{lem:symplectic-reduction}
For every $S\in\Sp(V_A\oplus V_B)$, there are symplectic transformations $R_1,\ldots,R_N\in\Sp(V_A\oplus V_B)$, each representing an elementary block as defined in Sec.~\ref{sec:decomposition}, such that
\begin{equation}
  S=R_1R_2\cdots R_N
  \label{eq:symplectic-factorization}
\end{equation}
and
\begin{equation}
  \sum_{j=1}^{N}\rank R_j^{BA}=\rank S^{BA}.
  \label{eq:rank-telescope}
\end{equation}
\end{lemma}

\begin{proof}
We first show that the sum of the ranks of the factors' $BA$ blocks is bounded below by $\rank S^{BA}$, then construct a factorization that attains this bound.
For $Q_1,Q_2\in\Sp(V_A\oplus V_B)$, block multiplication gives
\begin{equation}
  (Q_1Q_2)^{BA}=Q_1^{BA}Q_2^{AA}+Q_1^{BB}Q_2^{BA},
\end{equation}
so the rank inequalities for sums and products imply
\begin{equation}
  \rank (Q_1Q_2)^{BA}\leq\rank Q_1^{BA}+\rank Q_2^{BA}.
  \label{eq:cross-rank-subadditivity}
\end{equation}
Applying this inequality repeatedly to any factorization $S=R_1\cdots R_N$ gives
\begin{equation}
  \rank S^{BA}=\rank(R_1\cdots R_N)^{BA}
  \leq\sum_{j=1}^{N}\rank R_j^{BA}.
  \label{eq:factorization-rank-lower-bound}
\end{equation}
To prove Eq.~\eqref{eq:rank-telescope}, it therefore suffices to construct a factorization $S=R_1\cdots R_N$ satisfying the reverse inequality,
\begin{equation}
  \sum_{j=1}^{N}\rank R_j^{BA}\leq\rank S^{BA}.
  \label{eq:factorization-rank-upper-bound}
\end{equation}
The two inequalities then force equality.
We construct the factors iteratively, starting from $S_0=S$.
At step $i\geq1$, we stop if $\rank S_{i-1}^{BA}=0$.
Otherwise, we construct a symplectic transformation $R_i\in\Sp(V_A\oplus V_B)$ representing an elementary block, with $\rank R_i^{BA}\in\{1,2\}$, and set
\begin{equation}
  S_i=R_i^{-1}S_{i-1}\in\Sp(V_A\oplus V_B)
  \label{eq:iterative-residual}
\end{equation}
so that
\begin{equation}
  \rank S_i^{BA}\leq\rank S_{i-1}^{BA}-\rank R_i^{BA}.
  \label{eq:rank-reduction-target}
\end{equation}
Each $S_i$ is symplectic because $R_i$ and $S_{i-1}$ are symplectic.
Equality holds in Eq.~\eqref{eq:rank-reduction-target}, since applying Eq.~\eqref{eq:cross-rank-subadditivity} to $S_{i-1}=R_iS_i$ gives the reverse inequality.
To perform such a step, set
\begin{equation}
  W=S_{i-1}(V_A),\qquad W_A=W\cap V_A,\qquad W_B=W\cap V_B.
\end{equation}
Thus $W_A$ and $W_B$ consist of the vectors in $W$ supported only on Alice and only on Bob, respectively.
Define the coordinate projection $\pi_B:V_A\oplus V_B\to V_B$ by
\begin{equation}
  \pi_B(a+b)=b,\qquad a\in V_A,\quad b\in V_B.
\end{equation}
Its restriction $\pi_B|_W:W\to V_B$ sends a vector to zero exactly when that vector lies in $V_A$.
Hence
\begin{equation}
  \ker(\pi_B|_W)=W\cap V_A=W_A.
\end{equation}
Moreover, every vector of $W$ is $S_{i-1}x$ for some $x\in V_A$, and $\pi_B(S_{i-1}x)=S_{i-1}^{BA}x$.
Therefore
\begin{equation}
  \operatorname{im}(\pi_B|_W)=\operatorname{im}S_{i-1}^{BA}.
\end{equation}
Applying the rank--nullity theorem to $\pi_B|_W:W\to V_B$, and using $\dim W=\dim V_A$ because $S_{i-1}$ is invertible, we obtain
\begin{align}
  \rank S_{i-1}^{BA}
  &=\rank(\pi_B|_W)\notag\\
  &=\dim W-\dim\ker(\pi_B|_W)\notag\\
  &=\dim V_A-\dim W_A.
  \label{eq:rank-from-W_A}
\end{align}

The update in Eq.~\eqref{eq:iterative-residual} gives $S_i(V_A)=R_i^{-1}(W)$.
As in Eq.~\eqref{eq:rank-from-W_A}, applying the rank--nullity theorem to the restriction $\pi_B|_{R_i^{-1}(W)}:R_i^{-1}(W)\to V_B$ gives
\begin{equation}
  \rank S_i^{BA}=\dim V_A-\dim(R_i^{-1}(W)\cap V_A).
\end{equation}
Together with Eq.~\eqref{eq:rank-from-W_A}, this shows that Eq.~\eqref{eq:rank-reduction-target} is equivalent to
\begin{equation}
  \dim(R_i^{-1}(W)\cap V_A)\geq\dim W_A+\rank R_i^{BA}.
  \label{eq:intersection-growth-target}
\end{equation}
Since $W_A,W_B\subseteq W$ and $V_A\cap V_B=\{0\}$, we have $W_A\oplus W_B\subseteq W$.
We construct $R_i$ satisfying Eq.~\eqref{eq:intersection-growth-target} by distinguishing whether this inclusion is strict.
\begin{enumerate}
\item Suppose $W\neq W_A\oplus W_B$.
In this case, we construct a symplectic transformation $R_i$ representing a CZ-type elementary block.

Choose $a+b\in W\setminus(W_A\oplus W_B)$ with $a\in V_A$ and $b\in V_B$.
If $a\in W_A$, then $b=(a+b)+a\in W\cap V_B=W_B$, contradicting this choice.
Similarly, $b\in W_B$ would imply $a\in W_A$.
Thus $a\notin W_A$ and $b\notin W_B$, so both are nonzero.
If $a$ were orthogonal to every vector of $W_A^{\perp_A}$, Lemma~\ref{lem:orthogonal-complement} would give $a\in(W_A^{\perp_A})^{\perp_A}=W_A$, again a contradiction.
We can therefore choose $q\in W_A^{\perp_A}$ such that
\begin{equation}
  [a,q]=1,\qquad [k,q]=0\quad(k\in W_A).
\end{equation}
Apply Lemma~\ref{lem:cz-symplectic} to $q$ and $b$ to obtain a CZ-type symplectic transformation $R_i$ with $\rank R_i^{BA}=1$.
Equation~\eqref{eq:cz-linear-action} gives
\begin{equation}
  R_i(k)=k\quad(k\in W_A),\qquad R_i(a)=a+b\in W.
\end{equation}

To prove Eq.~\eqref{eq:intersection-growth-target}, we show that $R_i^{-1}(W)\cap V_A$ contains $W_A$ together with the additional vector $a$.
Since $R_i(k)=k$ for $k\in W_A$ and $R_i(a)=a+b$, invertibility gives
\begin{equation}
  R_i^{-1}(k)=k\quad(k\in W_A),\qquad R_i^{-1}(a+b)=a.
\end{equation}
For every $k\in W_A$, the vectors $k$ and $a+b$ belong to $W$, so their inverse images $k$ and $a$ belong to $R_i^{-1}(W)$.
Both also lie in $V_A$.
The intersection $R_i^{-1}(W)\cap V_A$ is a linear subspace, so it contains all their linear combinations:
\begin{equation}
  W_A\oplus\operatorname{span}\{a\}\subseteq R_i^{-1}(W)\cap V_A.
\end{equation}
The sum is direct because $a\notin W_A$, so its dimension is $\dim W_A+1=\dim W_A+\rank R_i^{BA}$, proving Eq.~\eqref{eq:intersection-growth-target}.

\item Suppose $W=W_A\oplus W_B$.
In this case, we construct a symplectic transformation $R_i$ representing a SWAP-type elementary block.

We first establish the orthogonal decompositions
\begin{equation}
  V_A=W_A\oplus W_A^{\perp_A},\qquad W=W_A\oplus W_B,
\end{equation}
in which all summands are nondegenerate, together with the dimension relation
\begin{equation}
  \dim W_A^{\perp_A}=\dim W_B=\rank S_{i-1}^{BA}>0.
  \label{eq:dimension-relation}
\end{equation}

To verify the nondegeneracy assertions, first note that $W=S_{i-1}(V_A)$ is nondegenerate because $V_A$ is nondegenerate and $S_{i-1}$ preserves the symplectic form.
The assumed decomposition $W=W_A\oplus W_B$ is orthogonal because $V_A\perp V_B$.
If a vector in $W_A$ is orthogonal to all of $W_A$, it is also orthogonal to $W_B$, hence to all of $W$, and must therefore be zero.
The same argument applies to $W_B$, so both $W_A$ and $W_B$ are nondegenerate.
Lemma~\ref{lem:orthogonal-complement} now gives $V_A=W_A\oplus W_A^{\perp_A}$ and the nondegeneracy of $W_A^{\perp_A}$.

To establish the dimension relation in Eq.~\eqref{eq:dimension-relation}, note that invertibility of $S_{i-1}$ gives $\dim W=\dim V_A$.
Taking dimensions in the two decompositions therefore gives $\dim W_A^{\perp_A}=\dim W_B=\dim V_A-\dim W_A$.
By Eq.~\eqref{eq:rank-from-W_A}, this common dimension is $\rank S_{i-1}^{BA}$, which is positive by assumption.

Choose symplectic pairs $(p_1,p_2)$ in $W_A^{\perp_A}$ and $(h_1,h_2)$ in $W_B$, with $[p_1,p_2]=[h_1,h_2]=1$.
Such pairs exist because $W_A^{\perp_A}$ and $W_B$ are nonzero and nondegenerate.
Write $V_A'=\operatorname{span}\{p_1,p_2\}\subseteq V_A$ and $V_B'=\operatorname{span}\{h_1,h_2\}\subseteq V_B$.

Apply Lemma~\ref{lem:swap-symplectic} to the pairs $(p_1,p_2)$ and $(h_1,h_2)$ to obtain a SWAP-type symplectic transformation $R_i$ with $\rank R_i^{BA}=2$.

Finally, to prove Eq.~\eqref{eq:intersection-growth-target}, we show that $R_i^{-1}(W)\cap V_A$ contains the direct sum $W_A\oplus V_A'$.
First, the inclusion $V_A'\subseteq W_A^{\perp_A}$ gives $W_A\subseteq (V_A')^{\perp_A}$, so $R_i$ fixes every vector in $W_A$.
Since $W_A\subseteq W$, this implies $W_A\subseteq R_i^{-1}(W)\cap V_A$.
Next, $R_i(V_A')=V_B'\subseteq W_B\subseteq W$, so $V_A'\subseteq R_i^{-1}(W)\cap V_A$ as well.
These two subspaces intersect trivially: $W_A\cap V_A'\subseteq W_A\cap W_A^{\perp_A}=\{0\}$ because $W_A$ is nondegenerate.
Therefore
\begin{equation}
  W_A\oplus V_A'\subseteq R_i^{-1}(W)\cap V_A.
\end{equation}
The subspace on the left has dimension $\dim W_A+2=\dim W_A+\rank R_i^{BA}$, proving Eq.~\eqref{eq:intersection-growth-target}.
\end{enumerate}

Thus Eq.~\eqref{eq:rank-reduction-target} holds in both cases, with $\rank R_i^{BA}\in\{1,2\}$.
Each step therefore strictly decreases the nonnegative integer $\rank S_{i-1}^{BA}$, so the process stops after finitely many steps, say $M$, with $\rank S_M^{BA}=0$.
Rearranging Eq.~\eqref{eq:iterative-residual} as $S_{i-1}=R_iS_i$ and iterating gives
\begin{equation}
  S=R_1\cdots R_M S_M,
\end{equation}
with the product taken to be empty if $M=0$.
By Lemma~\ref{lem:block-diagonal-symplectic}, $S_M$ represents a local elementary block.
Set $R_{M+1}=S_M$ and $N=M+1$; then $S=R_1\cdots R_N$ and $\rank R_{M+1}^{BA}=0$.
If $\rank S^{BA}=0$ initially, we take $M=0$ and obtain $R_1=S$ directly.
Summing Eq.~\eqref{eq:rank-reduction-target} over all steps yields
\begin{align}
  \sum_{j=1}^{N}\rank R_j^{BA}
  &=\sum_{i=1}^{M}\rank R_i^{BA}\notag\\
  &\leq\sum_{i=1}^{M}\bigl(\rank S_{i-1}^{BA}-\rank S_i^{BA}\bigr)\notag\\
  &=\rank S_0^{BA}-\rank S_M^{BA}\notag\\
  &=\rank S^{BA}.
\end{align}
This proves Eq.~\eqref{eq:factorization-rank-upper-bound}; combining it with Eq.~\eqref{eq:factorization-rank-lower-bound} gives the required equality.
\end{proof}

\begin{proof}[Proof of Lemma~\ref{lem:unitary-decomposition}]
Apply Lemma~\ref{lem:symplectic-reduction} to the symplectic transformation $S\in\Sp(V_A\oplus V_B)$ of $U_c$, and choose a Clifford representative $G_j\in\cU(\cH_A\otimes\cH_B)$ of each $R_j$.
Any two such representatives differ by a Pauli up to phase, so each $G_j$ remains an elementary block.
The product $G=G_1\cdots G_N$ and $U_c$ have the same symplectic transformation.
Therefore
\begin{equation}
  G^\dagger U_c=e^{i\theta}(P_A\otimes P_B)
\end{equation}
for Pauli operators $P_A\in\cU(\cH_A)$ and $P_B\in\cU(\cH_B)$ and some $\theta\in\mathbb{R}$.
Append $P_A\otimes P_B$ on the right of $G_1\cdots G_N$ and relabel the elementary blocks.
\begin{samepage}
This additional local elementary block has operator Schmidt rank one and symplectic transformation $\mathbb{I}$, whose $BA$ block has rank zero.
Lemma~\ref{lem:opsch-rank} and Eq.~\eqref{eq:rank-telescope} then give
\begin{align}
  \sum_j\log_2\OpSch(G_j)
  &=\sum_j\rank R_j^{BA}\\
  &=\rank S^{BA}
  =\log_2\OpSch(U_c).\qedhere
\end{align}
\end{samepage}
\end{proof}

\section{Details of the example in Fig.~\ref{fig:hidden-swap-input}}
\label{app:hidden-swap-derivation}

We derive Eq.~\eqref{eq:hidden-swap-blocks} by following the procedure used to prove Lemma~\ref{lem:unitary-decomposition} in Appendix~\ref{app:decomposition}.

\subsection*{Symplectic matrix and operator Schmidt rank of the original unitary}

For each qubit $t\in\{a_1,a_2,b_1,b_2\}$, let $x_t,z_t\in V_A\oplus V_B$ denote the binary labels of $X_t,Z_t$, with identity operators on the other qubits.
We use the ordered basis
\begin{equation}
  (x_{a_1},x_{a_2},z_{a_1},z_{a_2},
    x_{b_1},x_{b_2},z_{b_1},z_{b_2}).
  \label{eq:example-symplectic-basis}
\end{equation}
The symplectic matrix of the unitary $U$ in Eq.~\eqref{eq:hidden-swap-blocks} is given by
\begin{equation}
  S_0=S=
  \partitionedmatrix{
    1\&1\&0\&0\&0\&1\&0\&0\\
    1\&1\&0\&0\&1\&0\&0\&0\\
    0\&0\&1\&0\&0\&0\&1\&0\\
    0\&0\&0\&0\&0\&0\&1\&0\\
    0\&1\&0\&0\&0\&0\&0\&0\\
    0\&1\&0\&0\&0\&1\&0\&0\\
    0\&0\&0\&1\&0\&0\&1\&1\\
    0\&0\&1\&0\&0\&0\&1\&1\\
  }.
  \label{eq:example-S0}
\end{equation}
The $BA$ block of $S_0$ has a zero first column and three linearly independent remaining columns, so
\begin{equation}
  \rank S_0^{BA}=3,\qquad \OpSch(U)=2^3=8,
\end{equation}
where the second equality follows from Lemma~\ref{lem:opsch-rank}.

For each residual matrix $S_i$, write
\begin{equation}
  \begin{aligned}
    W_i&=S_i(V_A),\\
    W_{i,A}&=W_i\cap V_A,\qquad W_{i,B}=W_i\cap V_B.
  \end{aligned}
\end{equation}

\subsection*{Extraction of a CZ-type elementary block}

From the first four columns of $S_0$, we obtain
\begin{align}
  W_0&=\operatorname{span}\{x_{a_1}+x_{a_2},\,x_{b_1}+x_{b_2},\,z_{a_1}+z_{b_2},\,z_{b_1}\},\notag\\
  W_{0,A}&=\operatorname{span}\{x_{a_1}+x_{a_2}\},\notag\\
  W_{0,B}&=\operatorname{span}\{x_{b_1}+x_{b_2},z_{b_1}\}.
\end{align}
Since $\dim W_0=4$ whereas $\dim W_{0,A}+\dim W_{0,B}=3$, we have $W_0\neq W_{0,A}\oplus W_{0,B}$.
This is Case 1 in the proof of Lemma~\ref{lem:symplectic-reduction}.
Choose
\begin{equation}
  a=z_{a_1},\qquad b=z_{b_1}+z_{b_2},\qquad q=x_{a_1}+x_{a_2}.
\end{equation}
Here, $a\in V_A$, $b\in V_B$, and $a+b=(z_{a_1}+z_{b_2})+z_{b_1}$ lies in $W_0\setminus(W_{0,A}\oplus W_{0,B})$.
Also, $q\in W_{0,A}^{\perp_A}$ and $[a,q]=1$.
Lemma~\ref{lem:cz-symplectic} therefore gives the CZ-type transformation
\begin{equation}
  \begin{aligned}
    R_1(v_A)&=v_A+[v_A,q]b &&(v_A\in V_A),\\
    R_1(v_B)&=v_B+[v_B,b]q &&(v_B\in V_B).
  \end{aligned}
  \label{eq:example-R1}
\end{equation}
The first residual is
\begin{equation}
  S_1=R_1^{-1}S_0=
  \partitionedmatrix{
    1\&1\&0\&0\&0\&0\&0\&0\\
    1\&1\&0\&0\&1\&1\&0\&0\\
    0\&0\&1\&0\&0\&0\&1\&0\\
    0\&0\&0\&0\&0\&0\&1\&0\\
    0\&1\&0\&0\&0\&0\&0\&0\\
    0\&1\&0\&0\&0\&1\&0\&0\\
    0\&0\&1\&1\&0\&0\&1\&1\\
    0\&0\&0\&0\&0\&0\&1\&1\\
  },
  \label{eq:example-S1}
\end{equation}
with $\rank S_1^{BA}=2$.

\subsection*{Extraction of a SWAP-type elementary block}

The first four columns of $S_1$ give
\begin{align}
  W_{1,A}&=\operatorname{span}\{x_{a_1}+x_{a_2},z_{a_1}\},\notag\\
  W_{1,B}&=\operatorname{span}\{x_{b_1}+x_{b_2},z_{b_1}\},\notag\\
  W_1&=W_{1,A}\oplus W_{1,B}.
\end{align}
This is Case 2 in the proof of Lemma~\ref{lem:symplectic-reduction}.
Here,
\begin{equation}
  W_{1,A}^{\perp_A}
  =\operatorname{span}\{z_{a_1}+z_{a_2},x_{a_2}\}.
\end{equation}
Choose the symplectic pairs
\begin{align}
  (p_1,p_2)&=(z_{a_1}+z_{a_2},x_{a_2})
    &&\text{in }W_{1,A}^{\perp_A},\notag\\
  (h_1,h_2)&=(z_{b_1},x_{b_1}+x_{b_2})
    &&\text{in }W_{1,B}.
\end{align}
They satisfy $[p_1,p_2]=[h_1,h_2]=1$.
The SWAP-type transformation $R_2$ from Lemma~\ref{lem:swap-symplectic} exchanges these pairs:
\begin{equation}
  R_2(p_j)=h_j,\qquad R_2(h_j)=p_j\quad(j=1,2).
  \label{eq:example-R2}
\end{equation}
It fixes $q,a,x_{b_2},b$, which span the symplectic orthogonal complement of $\operatorname{span}\{p_1,p_2,h_1,h_2\}$.
The second residual is
\begin{equation}
  S_2=R_2^{-1}S_1=
  \partitionedmatrix{
    1\&1\&0\&0\&0\&0\&0\&0\\
    1\&0\&0\&0\&0\&0\&0\&0\\
    0\&0\&0\&1\&0\&0\&0\&0\\
    0\&0\&1\&1\&0\&0\&0\&0\\
    0\&0\&0\&0\&1\&1\&0\&0\\
    0\&0\&0\&0\&1\&0\&0\&0\\
    0\&0\&0\&0\&0\&0\&0\&1\\
    0\&0\&0\&0\&0\&0\&1\&1\\
  }.
  \label{eq:example-S2}
\end{equation}
Now $\rank S_2^{BA}=0$, and
\begin{equation}
  W_2=W_{2,A}=V_A,\qquad W_{2,B}=\{0\}.
\end{equation}
The iteration stops, leaving the local symplectic transformation $S_2$ and the factorization $S_0=R_1R_2S_2$.
The two extracted blocks have $BA$ ranks one and two, whose sum equals $\rank S_0^{BA}=3$.

\subsection*{Recovery of the unitary decomposition}

To identify the Clifford gates corresponding to $R_1,R_2$, recall the local Clifford unitary $D$ from Eq.~\eqref{eq:hidden-swap-local-factors}:
\begin{equation}
  D=\bigl(\mathrm{CNOT}_{a_1a_2}H_{a_1}\bigr)
    \otimes\mathrm{CNOT}_{b_1b_2}.
\end{equation}
Conjugation by $D$ sends the Pauli labels $z_{a_1},z_{b_2}$ to $q,b$, respectively.
It also sends $(z_{a_2},x_{a_2})$ to $(p_1,p_2)$ and $(z_{b_1},x_{b_1})$ to $(h_1,h_2)$.
Therefore, Clifford representatives of $R_1$ and $R_2$ are
\begin{equation}
  G_1=D\CZ_{a_1b_2}D^\dagger,\qquad
  G_2=D\SWAP_{a_2b_1}D^\dagger,
\end{equation}
respectively.
The block-diagonal matrix $S_2$ is represented by the local Clifford unitary $G_3=F$, where
\begin{equation}
  F=\bigl(\mathrm{CNOT}_{a_1a_2}\mathrm{CNOT}_{a_2a_1}\bigr)
    \otimes\bigl(\mathrm{CNOT}_{b_1b_2}\mathrm{CNOT}_{b_2b_1}\bigr).
\end{equation}
Thus $G_1G_2G_3$ and $U$ have the same symplectic matrix.

Direct calculation shows that $U$ and $G_1G_2G_3$ act identically, including the phase, on $\lvert\alpha_1\alpha_2\rangle_A\otimes\lvert\beta_1\beta_2\rangle_B$ for all $\alpha_i,\beta_i\in\F_2$ ($i=1,2$).
Since these states form a basis, we obtain
\begin{equation}
  U=G_1G_2G_3.
  \label{eq:example-unitary-recovery}
\end{equation}

\subsection*{Entanglement costs of the previous methods}

We calculate the Bell-pair consumption for the two previous methods listed in Table~\ref{tab:hidden-swap-comparison}.
We apply the gate-grouping procedures of Refs.~\cite{andresmartinez2019,wu2023} to the circuit in Fig.~\ref{fig:hidden-swap-input}, with the qubit allocation fixed throughout.
After converting each CNOT gate using $\mathrm{CNOT}_{uv}=H_v\CZ_{uv}H_v$, label the five CZ gates $Q_1,\ldots,Q_5\in\cU(\cH_A\otimes\cH_B)$ in execution order from left to right:
\begin{equation}
  \begin{aligned}
    Q_1&=\CZ_{a_1b_1},\qquad Q_2=\CZ_{a_2b_2},\\
    Q_3&=\CZ_{b_1a_2},\qquad Q_4=\CZ_{a_2b_1},\\
    Q_5&=\CZ_{b_2a_1}.
  \end{aligned}
\end{equation}
The CZ gates encountered along each wire are
\begin{equation}
  \begin{aligned}
    a_1 &: (Q_1,Q_5), & a_2 &: (Q_2,Q_3,Q_4),\\
    b_1 &: (Q_1,Q_3,Q_4), & b_2 &: (Q_2,Q_5).
  \end{aligned}
\end{equation}
Every consecutive pair in these lists is separated by a Hadamard gate on that wire.
Thus, in the hypergraph construction in Sec.~III~A of Ref.~\cite{andresmartinez2019}, no two CZ gates belong to the same group.
Each of the five CZ gates crosses the Alice--Bob partition and contributes one hyperedge cut, giving a cost of five Bell pairs.

For the method of Ref.~\cite{wu2023}, the block between $Q_2$ and $Q_4$ on $a_2$ is $H_{a_2}Q_3H_{a_2}$, which satisfies the global H-type embedding rule in Corollary~30 of that reference.
Consequently, $Q_2$ and $Q_4$ form a single packet, meaning that their distributed implementation shares one Bell pair.
The embedded gate $Q_3$ still consumes a separate Bell pair, while the embedding requires only additional local corrections.
Together with the single-gate packets for $Q_1$ and $Q_5$, this gives the four packets $\{Q_2,Q_4\}$, $\{Q_1\}$, $\{Q_3\}$, and $\{Q_5\}$, each consuming one Bell pair.

\section{Proof of Theorem~\ref{thm:clifford-t-constant-space}}
\label{app:clifford-t}

\begin{proof}
For a signed Pauli operator $P\in\cU(\cH_A\otimes\cH_B)$, meaning a tensor product of $\mathbb{I},X,Y,Z$ multiplied by $+1$ or $-1$, define
\begin{equation}
  R(P):=\exp\left(-\frac{i\pi}{8}P\right).
  \label{eq:Pauli-rotation}
\end{equation}
The $T$ gate on qubit $q_j$ can be written as $T_{q_j}=e^{i\pi/8}R(Z_{q_j})$.
Substituting into Eq.~\eqref{eq:exact-Clifford-T-decomposition} yields
\begin{equation}
  U=e^{i\phi'}C_tR(Z_{q_t})C_{t-1}R(Z_{q_{t-1}})\cdots
  C_1R(Z_{q_1})C_0,
\end{equation}
where $\phi':=\phi+t\pi/8$.
For every Clifford unitary $C\in\cU(\cH_A\otimes\cH_B)$, we have
\begin{equation}
  C R(P)=R(CPC^\dagger)C.
  \label{eq:commute-Clifford-through-rotation}
\end{equation}
Repeated application of this identity moves all Clifford factors to the right and gives
\begin{equation}
  U=e^{i\phi'}R(P_t)\cdots R(P_1)C.
  \label{eq:Pauli-rotation-Clifford-form}
\end{equation}
Here, $C:=C_tC_{t-1}\cdots C_1C_0$ is a Clifford unitary.
For $j=1,\ldots,t$, $P_j$ is the signed Pauli operator obtained by conjugating $Z_{q_j}$ by $C_t\cdots C_j$.

We next show that $R(P)$ can be implemented using at most one Bell pair with qubit overhead at most $(1,1)$.
Write $P=\epsilon P_A\otimes P_B$, where $\epsilon\in\{+1,-1\}$, and $P_A\in\cU(\cH_A)$ and $P_B\in\cU(\cH_B)$ are tensor products of $\mathbb{I},X,Y,Z$.
If $P_A=\mathbb{I}_A$, then $R(P)=\mathbb{I}_A\otimes\exp(-i\epsilon\pi P_B/8)$, and analogously $R(P)$ is local if $P_B=\mathbb{I}_B$.
In either case, $R(P)$ has operator Schmidt rank one and requires neither shared entanglement nor auxiliary qubits.
Henceforth assume $P_A\neq\mathbb{I}_A$ and $P_B\neq\mathbb{I}_B$.
Since $P^2=\mathbb{I}_A\otimes\mathbb{I}_B$, Eq.~\eqref{eq:Pauli-rotation} gives
\begin{equation}
  R(P)
  =\cos\left(\frac{\pi}{8}\right)\mathbb{I}_A\otimes\mathbb{I}_B
    -i\epsilon\sin\left(\frac{\pi}{8}\right)P_A\otimes P_B.
  \label{eq:Pauli-rotation-cosine-sine-expansion}
\end{equation}
To rewrite Eq.~\eqref{eq:Pauli-rotation-cosine-sine-expansion} in controlled-unitary form, let $\Pi_A^\pm:=(\mathbb{I}_A\pm P_A)/2\in\mathcal{L}(\cH_A)$ be the projectors onto the $+1$ and $-1$ eigenspaces of $P_A$, respectively.
Then $\mathbb{I}_A=\Pi_A^++\Pi_A^-$ and $P_A=\Pi_A^+-\Pi_A^-$.
Substituting these identities into Eq.~\eqref{eq:Pauli-rotation-cosine-sine-expansion} and using $P_B^2=\mathbb{I}_B$ gives
\begin{equation}
  R(P)=\Pi_A^+\otimes e^{-i\epsilon\pi P_B/8}
    +\Pi_A^-\otimes e^{i\epsilon\pi P_B/8}.
  \label{eq:Pauli-rotation-projector-expansion}
\end{equation}
Since $P_A$ is a nonidentity Pauli operator, there exists a Clifford unitary $U_A\in\cU(\cH_A)$ such that $U_A P_A U_A^\dagger=Z_a\otimes\mathbb{I}_{A\setminus a}$ for a qubit $a$ at Alice.
Here, $A\setminus a$ denotes Alice's remaining qubits.
Conjugating the projectors by $U_A$ gives
\begin{align}
  U_A\Pi_A^+U_A^\dagger
  &=\lvert0\rangle\!\langle0\rvert_a\otimes\mathbb{I}_{A\setminus a},\\
  U_A\Pi_A^-U_A^\dagger
  &=\lvert1\rangle\!\langle1\rvert_a\otimes\mathbb{I}_{A\setminus a}.
\end{align}
Consequently,
\begin{align}
  &(U_A\otimes\mathbb{I}_B)R(P)(U_A^\dagger\otimes\mathbb{I}_B)\notag\\
  &={\lvert0\rangle\!\langle0\rvert}_a\otimes\mathbb{I}_{A\setminus a}
    \otimes\exp(-i\epsilon\pi P_B/8)\notag\\
  &\quad{}+{\lvert1\rangle\!\langle1\rvert}_a\otimes\mathbb{I}_{A\setminus a}
    \otimes\exp(+i\epsilon\pi P_B/8),
  \label{eq:controlled-Pauli-rotation}
\end{align}
which applies $\exp(-i\epsilon\pi P_B/8)$ or $\exp(+i\epsilon\pi P_B/8)$ to Bob's qubits, controlled by the computational-basis value $0$ or $1$ of Alice's qubit $a$, respectively.
The controlled unitary in Eq.~\eqref{eq:controlled-Pauli-rotation} can be implemented using one Bell pair with qubit overhead $(1,1)$~\cite{eisert2000}.
Applying $U_A$ before this protocol and $U_A^\dagger$ afterward therefore implements $R(P)$ with the same resources.

Implement $C$ first using Theorem~\ref{thm:main}, and then implement $R(P_1),\ldots,R(P_t)$ in this order using the protocol above.
The local factors among $R(P_1),\ldots,R(P_t)$ require no shared entanglement, whereas each nonlocal factor $R(P_j)$ consumes one Bell pair.
Since there are at most $t$ nonlocal factors, the total Bell-pair consumption $K$ satisfies
\begin{equation}
  K\leq\log_2\OpSch(C)+t.
  \label{eq:cost-before-Clifford-bound}
\end{equation}

It remains to bound the Bell-pair cost of the Clifford factor $C$.
For this purpose, we use submultiplicativity of the operator Schmidt rank: for $M,N\in\mathcal{L}(\cH_A\otimes\cH_B)$,
\begin{equation}
  \OpSch(MN)\leq\OpSch(M)\OpSch(N).
  \label{eq:operator-Schmidt-submultiplicativity}
\end{equation}
Indeed, let $M=\sum_{j=1}^r A_j\otimes B_j$ and $N=\sum_{k=1}^{r'}A'_k\otimes B'_k$ be minimal operator Schmidt decompositions, so $r=\OpSch(M)$ and $r'=\OpSch(N)$.
Here, $A_j,A'_k\in\mathcal{L}(\cH_A)$ and $B_j,B'_k\in\mathcal{L}(\cH_B)$.
Then
\begin{equation}
  MN=\sum_{j=1}^r\sum_{k=1}^{r'}(A_jA'_k)\otimes(B_jB'_k).
\end{equation}
This product decomposition has at most $rr'$ terms, proving Eq.~\eqref{eq:operator-Schmidt-submultiplicativity}.
Equation~\eqref{eq:Pauli-rotation-Clifford-form} implies
\begin{equation}
  C=e^{-i\phi'}R(P_1)^\dagger\cdots R(P_t)^\dagger U.
\end{equation}
By Eq.~\eqref{eq:Pauli-rotation-cosine-sine-expansion}, each adjoint $R(P_j)^\dagger=R(-P_j)$ is a sum of at most two product operators and therefore has operator Schmidt rank at most two.
Consequently,
\begin{equation}
  \OpSch(C)\leq2^t\OpSch(U).
  \label{eq:collected-Clifford-Schmidt-bound}
\end{equation}
Since $C$ is Clifford, Lemma~\ref{lem:opsch-rank} shows that $\log_2\OpSch(C)$ is an integer.
Taking logarithms in Eq.~\eqref{eq:collected-Clifford-Schmidt-bound} shows that the integer $\log_2\OpSch(C)-t$ is at most $\log_2\OpSch(U)$.
Therefore,
\begin{equation}
  \log_2\OpSch(C)
  \leq\left\lfloor\log_2\OpSch(U)\right\rfloor+t.
  \label{eq:collected-Clifford-cost-bound}
\end{equation}

Combining this bound with Eq.~\eqref{eq:cost-before-Clifford-bound} gives
\begin{equation}
  K\leq\left\lfloor\log_2\OpSch(U)\right\rfloor+2t.
  \label{eq:constant-space-total-cost}
\end{equation}
The Clifford protocol has qubit overhead at most $(2,2)$, and each nonlocal $R(P_j)$ requires at most one auxiliary qubit at each party.
Reset and reuse the measured qubits between successive factors.
The complete protocol therefore has qubit overhead at most $(2,2)$.
\end{proof}

\bibliography{references}

\end{document}